\documentclass[12pt,draftcls,onecolumn]{IEEEtran}
\usepackage{etoolbox}
\makeatletter
\patchcmd{\@makecaption}
  {\scshape}
  {}
  {}
  {}
\makeatletter
\patchcmd{\@makecaption}
  {\\}
  {.\ }
  {}
  {}
\makeatother
 \usepackage{amsmath,amssymb}
 \usepackage{subfigure}
 \usepackage{graphicx,graphics,color,psfrag}
 \usepackage{cite,balance}
 \usepackage{caption}
 \allowdisplaybreaks
 \usepackage{algorithm}
 \usepackage{accents}
 \usepackage{amsthm}
 \usepackage{bm}
 \usepackage{algorithmic}
 \usepackage[english]{babel}
 \usepackage{multirow}
 \usepackage{enumerate}
 \usepackage{cases}
 \usepackage{stfloats}
 \usepackage{dsfont}
 \usepackage{color,soul}
 \usepackage{amsfonts}
 \usepackage{cite,graphicx,amsmath,amssymb}
 \usepackage{subfigure}
 \usepackage{fancyhdr}
 \usepackage{hhline}
 \usepackage{graphicx,graphics}
 \usepackage{array,color}
 \usepackage{amsmath}
\usepackage{float}
\usepackage{amssymb}
\usepackage{amsmath}
\usepackage{amsthm}
\usepackage{amsfonts}
\usepackage{graphicx}
\usepackage{algorithm}
\usepackage{algorithmic}
\usepackage{epstopdf}
\usepackage{cite}
\usepackage{amsmath,bm}
\usepackage{subfigure}
\usepackage{graphicx}
\usepackage{color}
\usepackage{graphicx}
\usepackage{calc}
\usepackage{caption}

\newtheorem{theorem}{\textbf{Theorem}}

\newtheorem{corollary}[theorem]{Corollary}

\newtheorem{lemma}{\textbf{Lemma}}

\newtheorem{problem}[theorem]{Problem}

\newtheorem{remark}{\textbf{Remark}}

\newtheorem{Def}{Definition}
\newtheorem{Prob}{\textbf{Problem}}

\begin{document}
\title{Joint Pushing and Caching for Bandwidth Utilization Maximization in Wireless Networks} %
\author{Yaping Sun, Ying Cui, and Hui Liu\\Shanghai Jiao Tong University, China\thanks{The paper was submitted in part to IEEE GLOBECOM 2017.}}
\maketitle

\begin{abstract}
Joint pushing and caching is recognized as an efficient remedy to the problem of spectrum scarcity incurred by tremendous mobile data traffic. In this paper, by exploiting storage resources at end users and predictability of user demand processes, we design the optimal joint pushing and caching policy to maximize bandwidth utilization, which is of fundamental importance to mobile telecom carriers. In particular, we formulate the stochastic optimization problem as an infinite horizon average cost Markov Decision Process (MDP), for which there generally exist only numerical solutions without many insights. By structural analysis, we show how the optimal policy achieves a balance between the current transmission cost and the future average transmission cost. In addition, we show that the optimal average transmission cost decreases with the cache size, revealing a tradeoff between the cache size and the bandwidth utilization. Then, due to the fact that obtaining a numerical optimal solution suffers the curse of dimensionality and implementing it requires a centralized controller and global system information, we develop a decentralized policy with polynomial complexity w.r.t. the numbers of users and files as well as cache size, by a linear approximation of the value function and optimization relaxation techniques. Next, we propose an online decentralized algorithm to implement the proposed low-complexity decentralized policy using the technique of Q-learning, when priori knowledge of user demand processes is not available. Finally, using numerical results, we demonstrate the advantage of the proposed solutions over some existing designs. The results in this paper offer useful guidelines for designing practical cache-enabled content-centric wireless networks.

\end{abstract}


\section{Introduction}
The rapid proliferation of smart mobile devices has triggered an unprecedented growth of global mobile data traffic \cite{Cisco}, resulting in the spectrum crunch problem in wireless systems. In order to improve the bandwidth utilization and support the sustainability of wireless systems, researchers have primarily focused on increasing the access rate of wireless systems and the density of network infrastructures, e.g., base stations (BSs). However, the expansive growth of both the access rate and the density of network infrastructures entails prohibitive network costs. On the other hand, modern data traffic exhibits a high degree of asynchronous content reuse \cite{Caire}. Thus, caching is gradually recognized as a promising approach to further improve the bandwidth utilization by placing contents closer to users, e.g., at BSs or even at end users, for future requests. Recent investigations show that caching can effectively reduce the traffic load of wireless and backhaul links as well as user-perceived latency \cite{users1,users2,Tassiulas,ES,WebCaching,CUI}.

Based on whether content placement is updated, caching policies can be divided into two categories, i.e., static caching policies and dynamic caching policies. Static caching policies refer to the caching policies under which content placement remains unchanged over a relatively long time. For example, \cite{Tassiulas,ES,CUI} consider static caching policies at BSs to reduce the traffic load of backhaul links. In addition, in \cite{users1,users2}, static caching policies at end users are proposed to not only alleviate the backhaul burden but also reduce the traffic load of wireless links. However, all the static caching policies in \cite{users1,users2,Tassiulas,ES,CUI} are designed based on content popularity, e.g., the probability of each file being requested, which is assumed to be known in advance, and cannot exploit temporal correlation of a demand process to further improve performance of cache-assisted systems.
Dynamic caching policies refer to the caching policies where content placement may update from time to time by making use of instantaneous user request information.
In this way, dynamic caching policies can not only operate without priori knowledge of content popularity but also capture the temporal correlation of a demand process. The least recently used (LRU) policy and the least frequently used (LFU) policy \cite{WebCaching} are two of the commonly adopted dynamic caching policies, primarily due to ease of implementation. 
However, they are both heuristic designs and may not guarantee promising performance in general.

Pure dynamic caching policies only focus on caching contents which have been requested and delivered to the local cache, and hence have limitations in smoothing traffic load fluctuations and enhancing bandwidth utilization. To address these limitations, joint pushing (i.e., proactively transmitting) and caching has been receiving more and more attention, as it can further improve bandwidth utilization. Specifically, the underutilized bandwidth at low traffic time can be exploited to proactively transmit contents for satisfying future user demands. Therefore, it is essential to design intelligent joint pushing and caching policy based on the knowledge of user demand processes.

For instance, \cite{SBS} considers joint pushing and caching to minimize the energy consumption, assuming complete knowledge of future content requests. In most cases, the assumption cannot be satisfied, and hence the proposed joint design has limited applications. To address this problem, \cite{Jsac,Hao,ChenWei,Yawei,Rice,JieGong} consider joint pushing and caching based on statistical information of content requests (e.g., content popularity), while \cite{Huang} considers online learning-aided joint design adaptive to instantaneous content requests and without priori knowledge of statistical information of content requests. Specifically, \cite{Hao} optimizes joint pushing and caching to maximize the network capacity in a push-based converged network with limited user storage. In \cite{JieGong}, the authors maximize the number of user requests served by small BSs (SBSs) via optimizing the pushing policy using Markov decision process (MDP). Note that in \cite{JieGong}, the cache size at each user is assumed to be unlimited, and thus caching design is not considered. In \cite{ChenWei} and \cite{Yawei}, the optimal joint pushing and caching policies are proposed to maximize the number of user requests served by the local caches in the scenarios of a single user and multiple users, respectively. \cite{Rice} studies the optimal joint pushing and caching policy to minimize the transmission cost. However, the joint designs in \cite{ChenWei,Yawei,Rice} do not take into account future reuse of requested files, and thus cannot be applied to certain applications which involve reusable contents,  such as music and video streaming. Moreover, in \cite{ChenWei,Yawei,Rice,Hao,JieGong}, temporal correlation of a demand process is not captured, and hence the potential of joint pushing and caching cannot be fully unleashed.
In contrast, \cite{Jsac} and \cite{Huang} exploit  temporal correlation in the joint designs. In particular, \cite{Jsac} investigates efficient transmission power control and caching to minimize both the access delay and the transmission cost using MDP. In \cite{Huang}, the authors maximize the average reward obtained by proactively serving user demands and propose an online learning-aided control algorithm. However, in \cite{Jsac} and \cite{Huang}, only a single user setup is considered without reflecting asynchronous demands for common contents from multiple users, and hence the proposed joint designs may not be directly applied to practical networks with multiple users. Moreover, the pushing policy in \cite{Huang} can predownload contents only one time slot ahead.

To further exploit the promises of joint pushing and caching in bandwidth utilization, in this paper, we investigate the optimal joint pushing and caching policy and reveal the fundamental impact of storage resource on bandwidth utilization. Specifically, we consider  a cache-enabled content-centric wireless  network consisting of a single server connected to multiple users via a shared and errorless link. Each user is equipped with a cache of limited size and generates inelastic file requests. We model the demand process of each user as a Markov chain, which captures both the asynchronous feature and temporal correlation of file requests. By the majorization theory \cite{majorization}, we choose a nondecreasing and strictly convex function of the traffic load as the per-stage cost and consider the time averaged transmission cost minimization. In particular, we formulate the joint pushing and caching optimization problem as an infinite horizon average cost MDP. Note that there generally exist only numerical solutions for MDPs, which suffer from the curse of dimensionality and cannot offer many design insights. Hence, it is a great challenge to design an efficient joint pushing and caching policy with acceptable complexity and offering design insights. In this paper, our main contributions are summarized as below.




\begin{itemize}
  \item First, we analyze structural properties of the optimal joint pushing and caching policy. In particular, by deriving an equivalent Bellman equation, we show that the optimal pushing policy balances the current transmission cost with the future average transmission cost, while the optimal caching policy achieves the lowest future average transmission cost given the optimal pushing policy. In addition, based on coupling and interchange arguments, we prove that the optimal average transmission cost decreases with the cache size, revealing the tradeoff between the cache size and the bandwidth utilization. Moreover, via relative value iteration, we analyze the partial monotonicity of the value function, based on which the sizes of both the state space and the caching action space are reduced, and thereby the complexity of computing the optimal joint design is reduced.
  \item Then, considering that obtaining the optimal policy requires computational complexity exponential with the number of users $K$ and combinatorial with the number of files $F$ as well as the cache size $C$, and implementing it requires a centralized controller and global system information, we develop a low-complexity (polynomial with $K$, $F$ and $C$) decentralized joint pushing and caching policy by using a linear approximation of the value function \cite{Cui1,Cui2} and optimization relaxation techniques.
  \item Next, noting that our proposed low-complexity decentralized policy requires statistic information of user demand processes, we propose an online decentralized algorithm (ODA) to implement the low-complexity decentralized policy using the technique of Q-learning \cite{bertsekas}, when priori knowledge of user demand processes is not available.
  \item Finally, by numerical results, we compare the performance of our proposed solutions with some existing designs at different system parameters, including the user number, file number, cache size and some key factors of user demand processes.
\end{itemize}
The key notations used in this paper are listed in Table I.

\section{System Model}
\begin{table}[t]
\caption{Key Notations}\label{tab:para_system}
\begin{center}
\vspace{-2mm}
\begin{tabular}{|c!{\vrule width 0.5pt}l|}
\hline
Notation&Meaning\rule{0pt}{3mm}\\
\hhline{|=|=|}
\hline
$\mathcal{F},\mathcal{K}$, $k,f$&  set of all files,  set of all users, user index,\ file index \\
\hline
$F$, $K$, $C$ &  file number, user number, cache size  \\

\hline
$\textbf{A}= (A_k)_{k\in \mathcal{K}}$, $\textbf{S}= (S_{k,f})_{k\in \mathcal{K},f\in \mathcal{F}}$&  system demand process, system cache state  \\
\hline
$X_k = (A_k,\textbf{S}_k)$, $\textbf{X}= (\textbf{A},\textbf{S})$ &  state of user $k$, system state \\
\hline
$\textbf{Q}_k=\big(q^{(k)}_{i,j}\big)_{i\in\bar{\mathcal{F}},j \in\bar{\mathcal{F}}}$ & transition matrix of demand process of user $k$ \\
\hline
$\textbf{R}= (R_f)_{f\in \mathcal{F}}$, $\textbf{P}= (P_f)_{f\in \mathcal{F}}$, $\Delta \textbf{S}= (\Delta \textbf{S}_k)_{k\in \mathcal{K}}$ &  reactive transmission action, pushing action, caching action \\
\hline
$\textbf{U}(\textbf{X})$, $\mu = (\mu_P, \mu_{\Delta S})$ &  system action space under $\textbf{X}$, joint pushing and caching policy  \\
\hline
$\theta$, $V(\textbf{X})$ &  optimal average cost, value function of system state $\textbf{X}$  \\
\hline

\end{tabular}
\vspace{-2mm}
\end{center}
\end{table}

\subsection{Network Architecture}
As in \cite{Ali}, we consider a cache-enabled content-centric wireless network with a single server connected through a shared error-free link to $K$ users,\footnote{Note that the server can be a BS and each user can be a mobile device or a SBS.} denoted as $\mathcal{K} \triangleq \{1,2,\cdots,K\}$, as shown in Fig.~\ref{fig1}.
\begin{figure}[t]
\begin{center}
 \includegraphics[width=5.5cm]{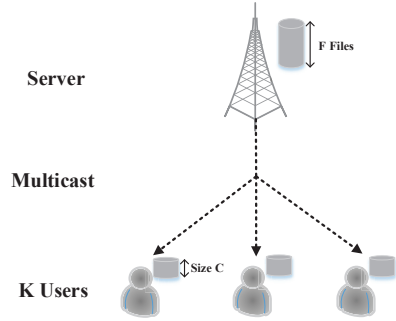}
\end{center}
 \caption{\small{System model.}
 }
\label{fig1}
\end{figure}
The server is accessible to a database of $F$ files, denoted as $\mathcal{F} \triangleq \{1,2,\cdots,F\}$. All the files are of the same size. Each user is equipped with a cache of size $C$ (in files). The system operates over an infinite time horizon and time is slotted, indexed by $t = 0,1,2,\cdots$. At the beginning of each time slot, each user submits at most one file request, which is assumed to be delay intolerant and must be served before the end of the slot, either by its own cache if the requested file has been stored locally, or by the server via the shared link. At each slot, the server can not only reactively transmit a file requested by some users at the slot but also push (i.e., proactively transmit) a file which has not been requested by any user at the slot. Each transmitted file can be received by all the users concurrently before the end of the time slot.\footnote{We assume that the duration of each time slot is long enough to average the small-scale channel fading process, and hence the ergodic capacity can be achieved using channel coding.} After being received, a file can be stored into some user caches.
\subsection{System State}
\subsubsection{Demand State}
At the beginning of time slot $t$, each user $k$ generates at most one file request. 
Let $A_k(t) \in \bar{\mathcal{F}}\triangleq\mathcal{F}\cup \{0\}$ denote the demand state of user $k$ at the beginning of time slot $t$, where $A_k(t)=0$ indicates that user $k$ requests nothing, and $A_k(t)=f\in \mathcal{F}$ indicates that user $k$ requests file $f$. Here, $\bar{\mathcal{F}}$ denotes the demand state space of each user which is of cardinality $F+1$. Let $\mathbf{A}(t)\triangleq(A_k(t))_{k\in\mathcal{K}} \in \bar{\mathcal{F}}^K$ denote the system demand state (of the $K$ users), where $\bar{\mathcal{F}}^K$ represents the system demand state space. Note that the cardinality of $\bar{\mathcal{F}}^K$ is $(F+1)^K $, which increases exponentially with $K$. 


For user $k$, we assume that $A_k(t)$ evolves according to a first-order $(F+1)$-state Markov chain, denoted as $\{A_k(t): t = 0,1,2,\cdots\}$, which captures temporal correlation of order one of user $k$'s demand process and is a widely adopted traffic model \cite{Huang}. 
Let $\Pr[A_k(t+1)=j|A_k(t)=i]$ denote the transition probability of going to state $j \in \bar{\mathcal{F}}$ at time slot $t+1$ given that the demand state at time slot $t$ is $i\in \bar{\mathcal{F}}$ for user $k$'s demand process. 
Assume that $\{A_k(t)\}$ is time-homogeneous and denote $q_{i,j}^{(k)} \triangleq \Pr[A_k(t+1)=j|A_k(t)=i]$.
Furthermore, we restrict our attention to an irreducible Markov chain. 
Denote with $\textbf{Q}_k\triangleq\big(q^{(k)}_{i,j}\big)_{i\in\bar{\mathcal{F}},j \in\bar{\mathcal{F}}}$ the transition probability matrix of $\{A_k(t)\}$. We assume that the $K$ time-homogeneous Markov chains, i.e., $\{A_k(t)\}$, $k\in \mathcal{K}$, are independent of each other. 
Thus, we have $\Pr[\mathbf{A}(t+1)=\mbox{\boldmath$j$}|\mathbf{A}(t)=\mbox{\boldmath$i$}]= \prod_{k=1}^{K} q^{(k)}_{i_k,j_k}$, where $\mbox{\boldmath$i$} \triangleq (i_k)_{k \in \mathcal{K}}\in \bar{\mathcal{F}}^K$ and $\mbox{\boldmath$j$} \triangleq (j_k)_{k \in \mathcal{K}} \in \bar{\mathcal{F}}^K$.



\subsubsection{Cache State}
Let $S_{k,f}(t)\in\{0,1\}$ denote the cache state of file $f$ in the storage of user $k$ at time slot $t$, where $S_{k,f}(t)=1$ means that file $f$ is cached in user $k$'s storage and $S_{k,f}(t)=0$ otherwise. Under the cache size constraint, we have
\begin{equation}\label{Storage Constraint}
  \sum_{f\in \mathcal{F}}S_{k,f}(t)\leq C,\ k \in \mathcal{K}.
\end{equation}
Let $\mathbf{S}_k(t) \triangleq (S_{k,f}(t))_{f\in\mathcal{F}} \in \mathcal{S}$ denote the cache state of user $k$ at time slot $t$, where $\mathcal{S} \triangleq \{(S_{f})_{f\in\mathcal{F}}\in\{0,1\}^F: \sum_{f\in \mathcal{F}}S_{f} \leq C\}$ represents the cache state space of each user. Here, the user index is suppressed considering that the cache state space is the same across all the users. 
Let $\mathbf{S}(t)\triangleq(S_{k,f}(t))_{k\in\mathcal{K},f\in\mathcal{F}}\in \mathcal{S}^K$ denote the system cache state at time slot $t$, where $ \mathcal{S}^K$ represents the system cache state space. The cardinality of $\mathcal{S}^K$ is
$\left(\sum_{i=0}^C\binom{F}{i}\right)^K$, which increases with the number of users $K$ exponentially.



\subsubsection{System State}
At time slot $t$, denote with $X_k(t) \triangleq (A_k(t),\mathbf{S}_k(t))\in \bar{\mathcal{F}}\times \mathcal{S} $ the state of user $k$, where $\bar{\mathcal{F}}\times \mathcal{S}$ represents the state space of user $k$. The system state consists of the system demand state and the system cache state, denoted as $\textbf{X}(t)\triangleq(\mathbf{A}(t),\textbf{S}(t))\in \bar{\mathcal{F}}^K\times \mathcal{S}^K$, where $\bar{\mathcal{F}}^K\times \mathcal{S}^K$ represents the system state space. Note that $\textbf{X}(t) = (X_k(t))_{k\in \mathcal{K}}$.


\subsection{System Action}

\subsubsection{Pushing Action}
A file transmission can be reactive or proactive at each time slot. Denote with $R_f(t)\in\{0,1\}$ the reactive transmission action for file $f$ at time slot $t$, where $R_f(t) = 1$ when there exists at least one user who requests file $f$ but cannot find it in its local cache and $R_f(t) = 0$ otherwise. Thus, we have
\begin{equation}\label{reactive}
 R_f(t) = \max_{k\in \mathcal{K}: A_k(t)=f}\  \big(1 - S_{k,f}(t)\big), \ f \in \mathcal{F},
\end{equation}
which is determined directly by $\textbf{X}(t)$.\footnote{Note that we do not need to design the reactive transmission action.} Denote with $\mathbf{R}(t)\triangleq (R_f(t))_{f\in \mathcal{F}}$ the system reactive transmission action at time slot $t$. Also, denote with $P_f(t)\in\{0,1\}$ the pushing action for file $f$ at time slot $t$, where $P_f(t)=1$ denotes that file $f$ is pushed (i.e., transmitted proactively) and $P_f(t) = 0$ otherwise.
Considering that file $f$ is transmitted at most once at time slot $t$, we have
\begin{equation}\label{once}
P_f(t) + R_f(t) \leq 1, \ f \in \mathcal{F},
\end{equation}
where $R_f(t)$ is given by (\ref{reactive}). Furthermore, if file $f$ has already been cached in each user's storage, there is no need to push it. Hence, we have
\begin{equation}\label{twice}
P_f(t) \leq 1-\min_{k\in \mathcal{K}} S_{k,f}(t), \ f \in \mathcal{F}.
\end{equation}
Denote with $\mathbf{P}(t) \triangleq (P_f(t))_{f\in \mathcal{F}} \in \mathbf{U}_P(\textbf{X}(t))$ the system pushing action at time slot $t$, where $ \mathbf{U}_P(\textbf{X})\triangleq\{(P_f)_{f\in\mathcal{F}}\in \{0,1\}^F: (\ref{once}),(\ref{twice})\}$ represents the system pushing action space under $\textbf{X}$.

System pushing action $\mathbf{P}$ together with reactive transmission action $\mathbf{R}$ incurs a certain transmission cost. We assume that the transmission cost is an increasing and continuously convex function of the corresponding traffic load, i.e., $\sum_{f\in \mathcal{F}} \big(R_f+P_f\big)$, denoted by $\phi(\cdot)$. In accordance with practice, we further assume that $\phi(0) = 0$. For example, we can choose $\phi(x)= a^x -1$ with $a > 1$ or $\phi(x)=x^d$ with $d\geq 2$.\footnote{Note that by choosing $\phi(x)=2^x-1$, $\phi(\sum_{f\in \mathcal{F}} R_f(t)+P_f(t))$ can represent the energy consumption at time slot $t$.} Here, we note that the per-stage transmission cost is bounded within set $\{0,\phi(1),\cdots,\phi(\min\{F,KC\})\}$. By the technique of majorization \cite{majorization}, a small time-averaged transmission cost with such a per-stage cost function corresponds to a small peak-to-average ratio of the bandwidth requirement, i.e., a high bandwidth utilization, which is of fundamental importance to a mobile telecom carrier, as illustrated in Fig.~$2$.
\begin{figure}[t]
\begin{center}
 \subfigure[\small{Sequence \{x(t)\}. $\frac{1}{5}\sum_{t=1}^5x(t)\! =\! 1$, and $\frac{1}{5}\sum_{t=1}^5x^2(t)\! = \!2.2$.}]
 {\resizebox{5.5cm}{!}{\includegraphics{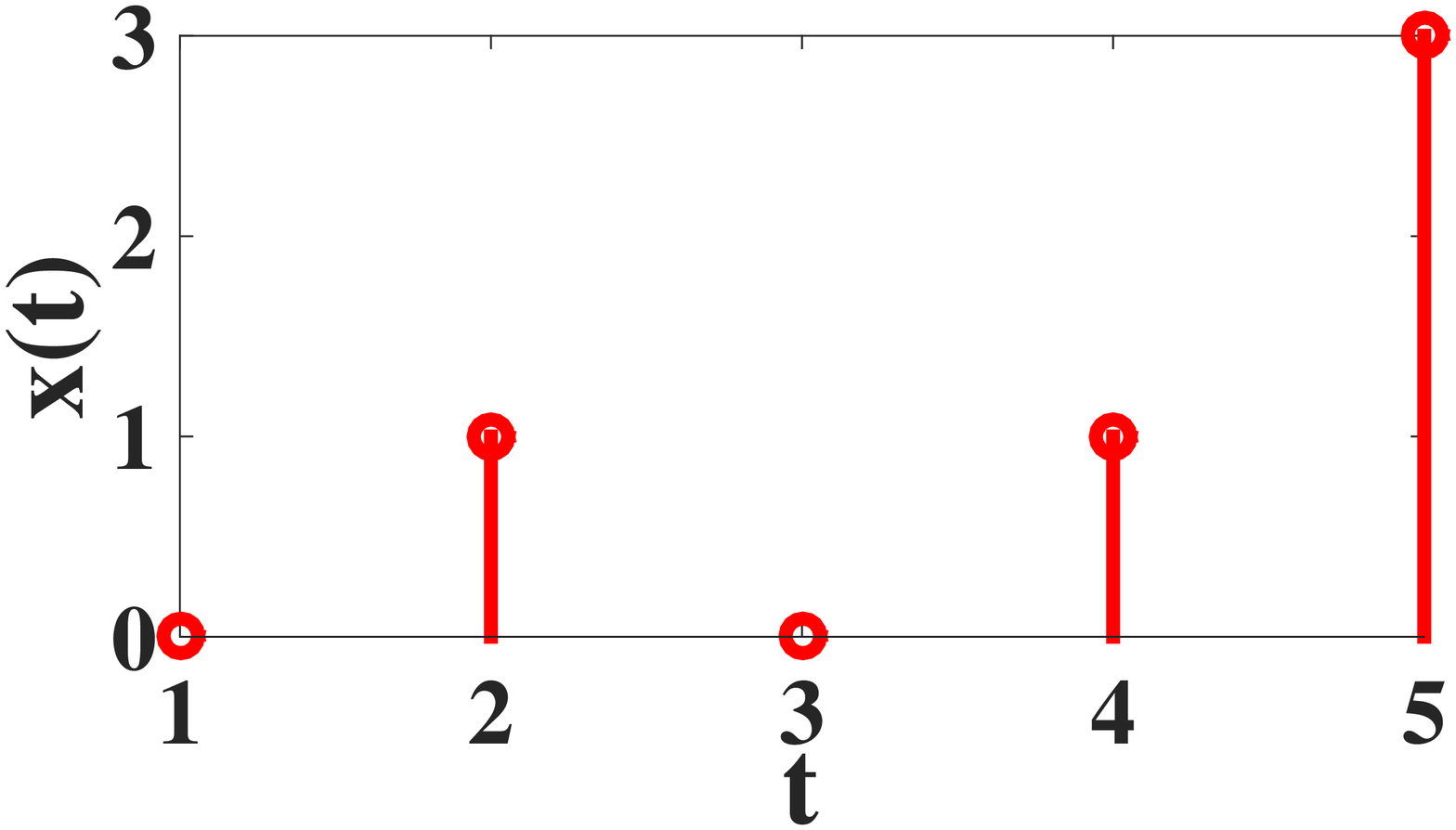}}}\quad\quad\quad\quad
 \subfigure[\small{Sequence \{y(t)\}. $\frac{1}{5}\sum_{t=1}^5y(t)\! = \!1$, and $\frac{1}{5}\sum_{t=1}^5y^2(t)\! = \!1$.}]
 {\resizebox{5.5cm}{!}{\includegraphics{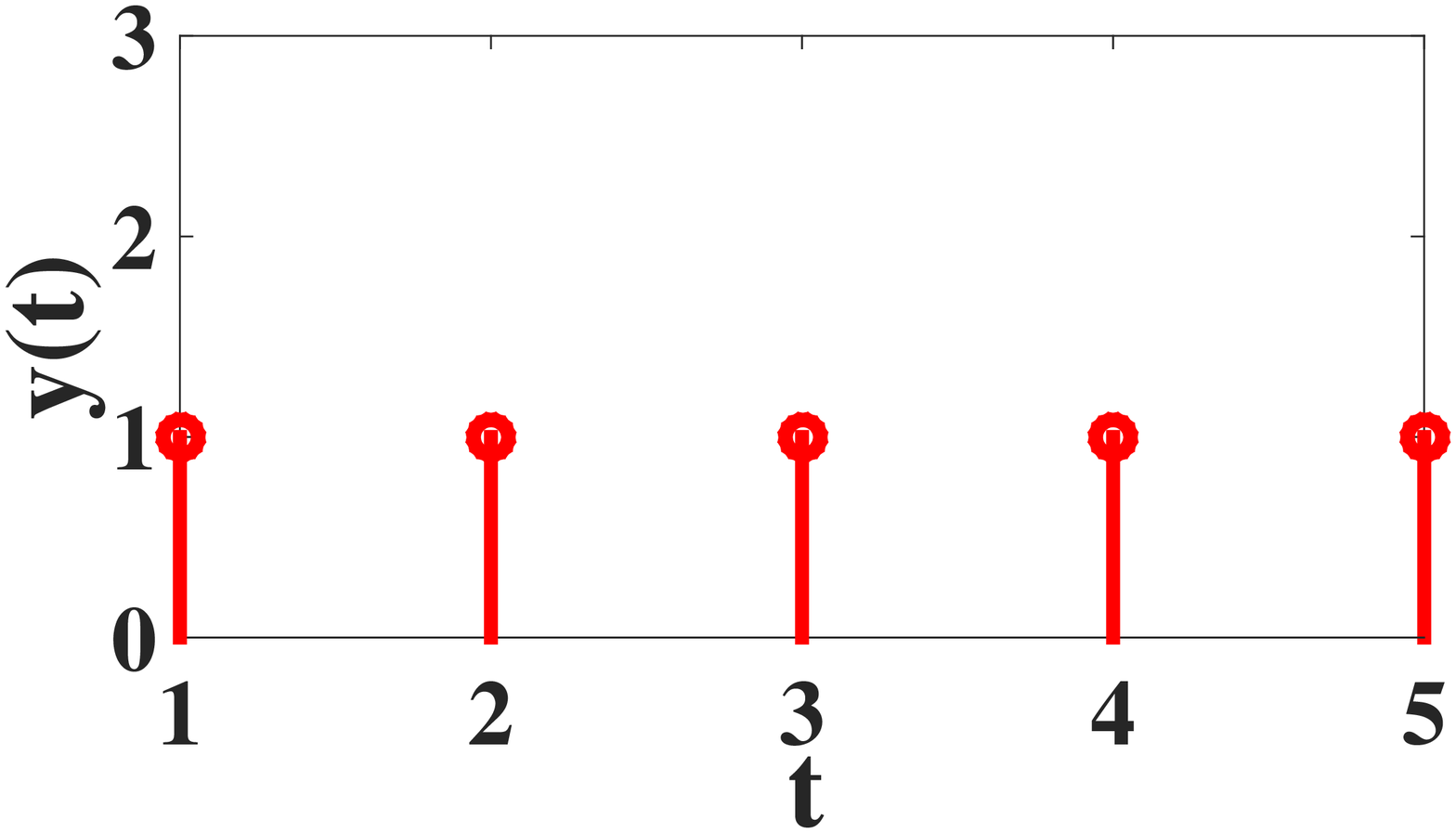}}}
 \end{center}
   \caption{\small{An illustration of the relationship between the average cost and bandwidth utilization. Note that $\frac{1}{5}\sum_{t=1}^5x(t)=\frac{1}{5}\sum_{t=1}^5y(t)$, while $\frac{1}{5}\sum_{t=1}^5x^2(t) >\frac{1}{5}\sum_{t=1}^5y^2(t)$}.} 
\label{Effect2}
\end{figure}
\subsubsection{Caching Action}
After the transmitted files being received by all the users, the system cache state can be updated. Let $\Delta S_{k,f}(t)\in \{-1,0,1\}$ denote the caching action for file $f$ at user $k$ at the end of time slot $t$, where $\Delta S_{k,f}(t) = 1$ means that file $f$ is stored into the cache of user $k$, $\Delta S_{k,f}(t) = 0$ implies that the cache state of file $f$ at user $k$ does not change, and $\Delta S_{k,f}(t) = -1$ indicates that file $f$ is removed from the cache of user $k$. Accordingly, the caching action satisfies the following cache update constraint:
\begin{equation}\label{Cache Decision Constraint}
  -S_{k,f}(t) \leq \Delta S_{k,f}(t) \leq R_f(t)+P_f(t),\ f\in\mathcal{F},\  k\in\mathcal{K},
\end{equation}
where $R_f(t)$ is given by (\ref{reactive}). In (\ref{Cache Decision Constraint}), the first inequality is to guarantee that file $f$ can be removed from the cache of user $k$ only when it has been stored at user $k$, and the second inequality is to guarantee that file $f$ can be stored into the cache of user $k$ only when it has been transmitted from the server.
The cache state evolves according to: 
\begin{equation}\label{Cache State Evolution}
  S_{k,f}(t+1)=S_{k,f}(t)+\Delta S_{k,f}(t),  \ \ \ \ f\in\mathcal{F},\  k\in\mathcal{K}.
\end{equation}
Since $S_{k,f}(t+1)$ belongs to $\{0,1\}$ and also satisfies (\ref{Storage Constraint}), we have the following two cache update constraints:
\begin{equation}\label{Storage Evolution Constraint2}
  S_{k,f}(t)+ \Delta S_{k,f}(t) \in \{0,1\}, \ \ \ \ \ f\in\mathcal{F},\ k\in\mathcal{K},
\end{equation}
\begin{equation}\label{Storage Evolution Constraint1}
\sum_{f\in \mathcal{F}}\big(S_{k,f}(t)+ \Delta S_{k,f}(t)\big)\leq C, \ \ \ f\in\mathcal{F},\ k\in\mathcal{K}.
\end{equation}
From (\ref{Cache Decision Constraint}),\! (\ref{Storage Evolution Constraint2})\! and \!(\ref{Storage Evolution Constraint1}), we
denote with $\Delta \mathbf{S}_k(t)$~$\triangleq$~$(\Delta S_{k,f}(t))_{f\in\mathcal{F}}$ $\!\in\! U_{\Delta S,k}(X_k(t),\mathbf{R}(t)+\mathbf{P}(t))$ the caching action of user $k$ at the end of time slot $t$, where $U_{\Delta S,k}(X_k,\!\mathbf{R}\!+\!\mathbf{P})\! \triangleq\! \{(\Delta S_{k,f})_{f\in\mathcal{F}}\!\in\! \{-1,0,1\}^F:\! (\ref{Cache Decision Constraint}),(\ref{Storage Evolution Constraint2}),(\ref{Storage Evolution Constraint1})\!\}$ represents \!the\! caching action space of user $k$ under its state $X_k$, system reactive transmission action $\textbf{R}$ and pushing action $\mathbf{P}$.
Let $\Delta \mathbf{S}(t)\!\triangleq\!(\Delta S_{k,f}(t))_{k\in \mathcal{K}, f\in \mathcal{F} }\!\in\!\mathbf{U}_{\Delta S}(\textbf{X}(t),\mathbf{P}(t))$ denote the system caching action at the end of time slot $t$, where $\mathbf{U}_{\Delta S}(\textbf{X},\mathbf{P})\!  \triangleq\! \prod_{k\in\mathcal K}\! U_{\Delta S,k}(X_k,\mathbf{R}\!+\!\mathbf{P}) $ represents the system caching action space under system state $\textbf{X}$ and pushing action~$\mathbf{P}$.

\subsubsection{System Action}
At time slot $t$, the system action consists of both the pushing action and caching action, denoted as $(\mathbf{P}(t),\Delta \textbf{S}(t)) \in \mathbf{U}(\textbf{X}(t))$, where $\mathbf{U}(\textbf{X}) \triangleq \{(\mathbf{P}, \Delta \textbf{S}): \Delta \textbf{S} \in \mathbf{U}_{\Delta S}(\textbf{X},\mathbf{P}),\ \mathbf{P}\in \mathbf{U}_P(\textbf{X})\}$ represents the system action space under system state $\textbf{X}$.

\section{Problem Formulation}
Given an observed system state $\textbf{X}$, the joint pushing and caching action, denoted as $ (\mathbf{P},\Delta \textbf{S})$, is determined according to a policy defined as below.

\begin{Def}[Stationary Joint Pushing and Caching Policy]
A stationary joint pushing and caching policy $\mu \triangleq (\mu_{P},\mu_{\Delta S})$ is a mapping from system state $\textbf{X}$ to system action $(\mathbf{P},\Delta \textbf{S})$, i.e., $ (\mathbf{P},\Delta \textbf{S}) = \mu(\textbf{X}) \in \mathbf{U}(\textbf{X})$. Specifically, we have $\mathbf{P} = \mu_{P}(\textbf{X})\in \mathbf{U}_P(\textbf{X})$ and $\Delta \textbf{S}=\mu_{\Delta S}(\textbf{X},\mathbf{P})\in \mathbf{U}_{\Delta S}(\textbf{X},\mathbf{P})$.
\end{Def}



From the properties of $\{\mathbf{A}(t)\}$ and $\{\textbf{S}(t)\}$,
 we see that the induced system state process $\{\textbf{X}(t)\}$ under policy $\mu$ is a controlled Markov chain.
%
%
The time averaged transmission cost under policy $\mu$ is given by
\begin{equation}\label{average transmission cost Objective1}
\bar{\phi}(\mu)\triangleq \limsup_{T\rightarrow \infty} \frac{1}{T}\sum_{t=0}^{T-1}\mathbb{E}\bigg[\phi\bigg(\sum_{f\in \mathcal{F}}\big(R_f(t)+P_f(t)\big)\bigg)\bigg],
\end{equation}
where $R_f(t)$ is given by (\ref{reactive}) and the expectation is taken w.r.t. the measure induced by the $K$ Markov chains. Note that $\bar{\phi}(\mu)$ can reflect the bandwidth utilization, as illustrated in Fig.~$2$.


In this paper, we aim to obtain an optimal joint pushing and caching policy $\mu$ to minimize the time averaged transmission cost $\bar{\phi}(\mu)$ defined in (\ref{average transmission cost Objective1}), i.e., maximizing the bandwidth utilization. Before formally introducing the problem, we first illustrate a simple example that highlights how the joint pushing and caching policy affects the average cost, i.e., bandwidth utilization.

\textbf{Motivating Example}. Consider a scenario with $K=4$, $F=4$, $C=1$ and $\phi(x)=x^2$. The user demand model is illustrated in Fig.~\ref{UserDemand}~(a). A sample path of the user demand processes is shown in Fig.~\ref{UserDemand}~(b). Note that at time slot $2$, there is no file request, while at time slot $3$, the number of file requests achieves the maximum value, i.e., $4$. Fig.~\ref{UserDemand} (c)-(h) illustrate the system cache states and the multicast transmission actions over three time slots under the following three policies: the most popular (MP) caching policy in which the $C$ most popular files (i.e., the first $C$ files with the maximum limiting probabilities) are cached at each user \cite{ES}, the LRU caching policy and a joint pushing and caching (JPC) policy. We can calculate the average cost over the three time slots under the aforementioned three policies, i.e., $\bar{\phi}_1\triangleq\frac{1^2+0^2+3^2}{3}= \frac{10}{3}$, $\bar{\phi}_2\triangleq\frac{1^2+0^2+4^2}{3}= \frac{17}{3}$ and $\bar{\phi}_3\triangleq\frac{1^2+1^2+1^2}{3}= 1$. Note that $\bar{\phi}_3<\bar{\phi}_1<\bar{\phi}_2$.
From Fig.~\ref{UserDemand} (h), we learn that under the joint pushing and caching policy, the bandwidth at low traffic time (e.g., time slot $2$) can be exploited to proactively transmit contents for satisfying future user demands (e.g., at time slot $3$), thereby improving the bandwidth utilization.

%
%

\begin{Prob}[Joint Pushing and Caching Optimization]
\begin{align}
& \bar{\phi}^* \triangleq \min_{\mu}~\ \bar{\phi}(\mu) \nonumber\\
& \ \ \ \ \ \ \ s.t. ~\ \ (\ref{reactive}),(\ref{once}),(\ref{twice}),(\ref{Cache Decision Constraint}),(\ref{Storage Evolution Constraint2}),(\ref{Storage Evolution Constraint1}),\nonumber
\end{align}\label{average transmission cost Minimization Problem}
where $\bar{\phi}^*$ denotes the minimum time averaged transmission cost under the optimal policy $\mu^*\triangleq(\mu^*_{P},\mu^*_{\Delta S})$, i.e., $\bar{\phi}^* = \bar{\phi}(\mu^*)$.
\end{Prob}

\begin{figure}[H]
\begin{center}
  \subfigure[Demand model for each user. $\textbf{Q}_k, k\in \mathcal{K}$ are the same.]
 {\resizebox{6cm}{!}{\includegraphics{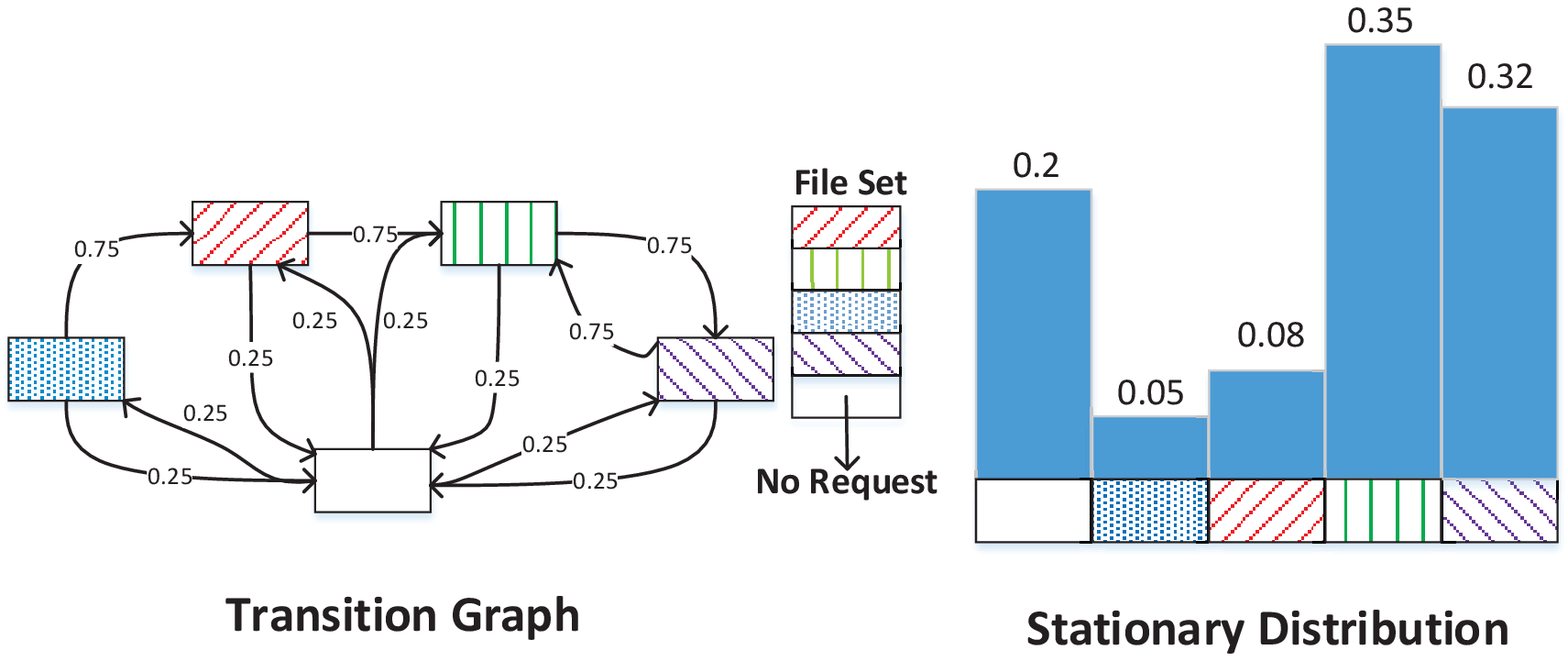}}}\quad\quad
 \subfigure[A\! sample\! path \!of\! $\{\textbf{A}(t)\}$.]
 {\resizebox{6cm}{!}{\includegraphics{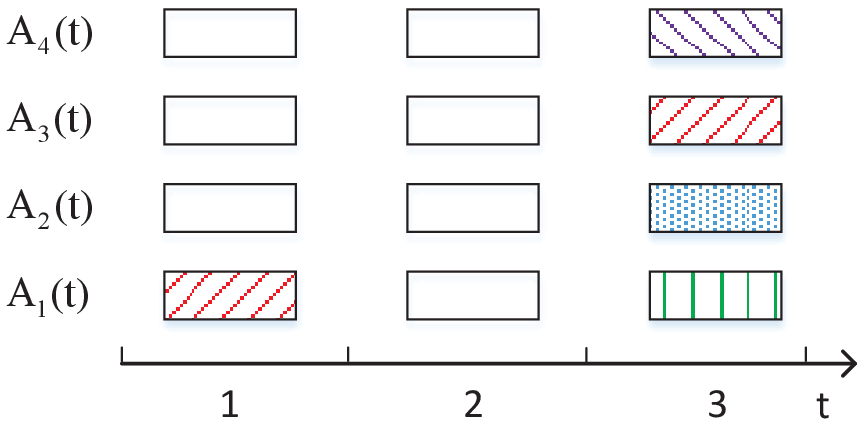}}}\quad\quad
 \subfigure[Cache state under MP.]
 {\resizebox{6cm}{!}{\includegraphics{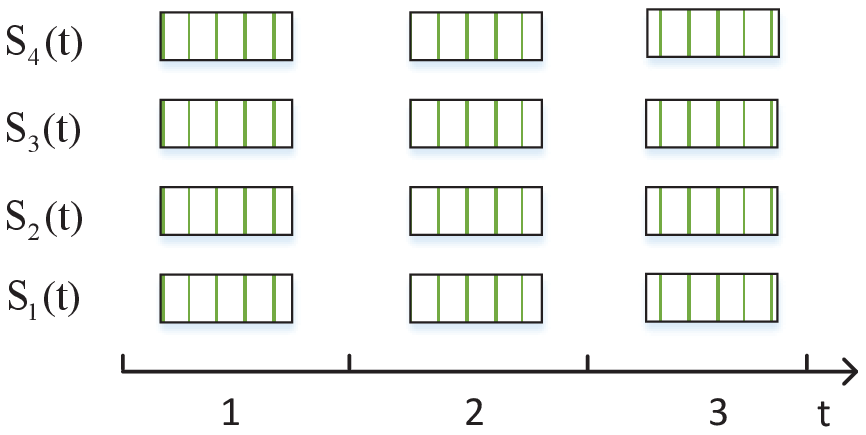}}}\quad\quad
 \subfigure[Reactive transmission under MP. Average cost $\bar{\phi}_1\triangleq\frac{1^2+0^2+3^2}{3}= \frac{10}{3}$.]
 {\resizebox{6cm}{!}{\includegraphics{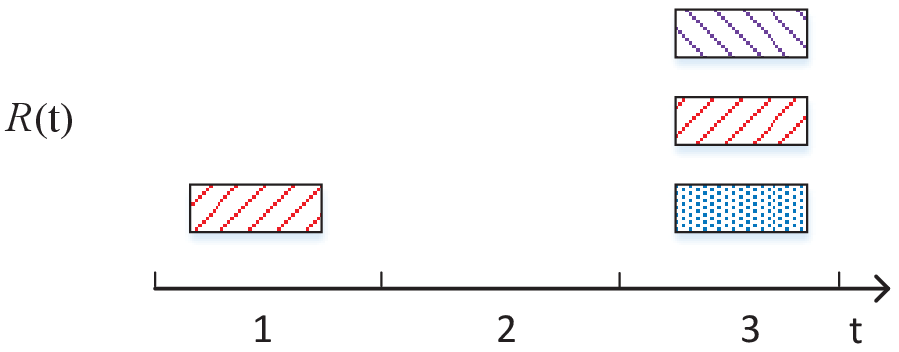}}}\quad\quad
  \subfigure[Cache state under LRU.]
 {\resizebox{6cm}{!}{\includegraphics{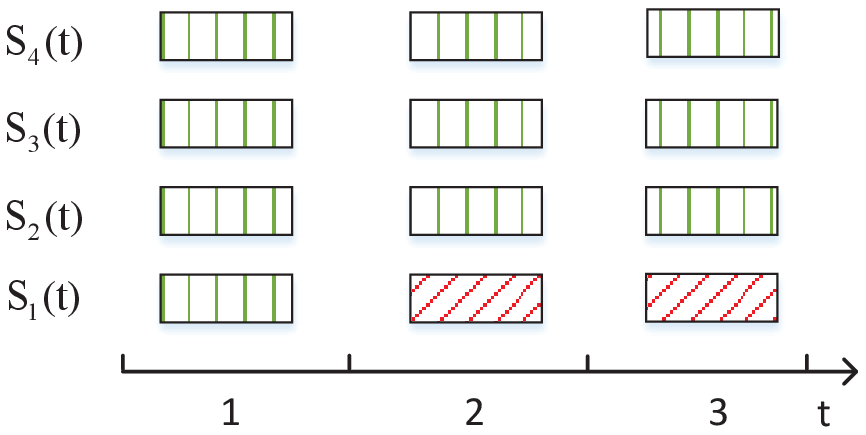}}}\quad\quad
 \subfigure[Reactive transmission under LRU. Average cost $\bar{\phi}_2\triangleq\frac{1^2+0^2+4^2}{3}= \frac{17}{3}$.]
 {\resizebox{6cm}{!}{\includegraphics{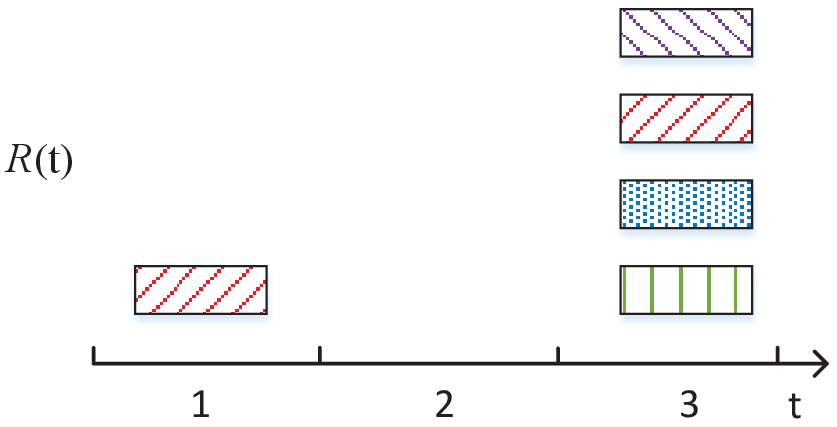}}}\quad\quad
   \subfigure[Cache state under JPC.]
 {\resizebox{6cm}{!}{\includegraphics{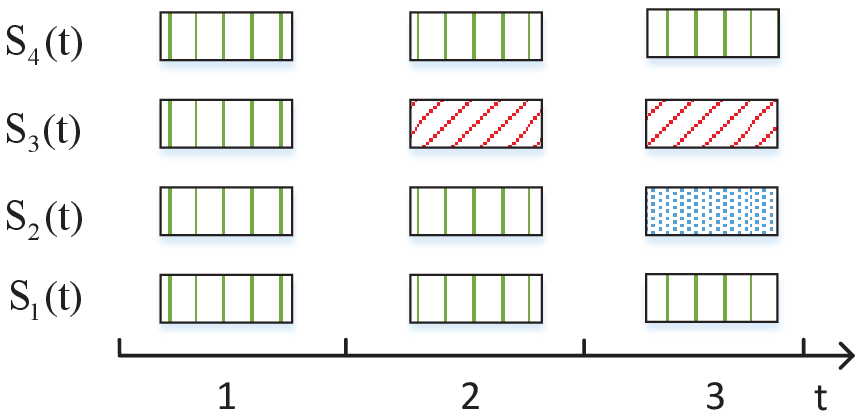}}}\quad\quad
 \subfigure[Reactive transmission and pushing under JPC. Average cost $\bar{\phi}_3\triangleq\frac{1^2+1^2+1^2}{3}= 1$.]
 {\resizebox{5.9cm}{!}{\includegraphics{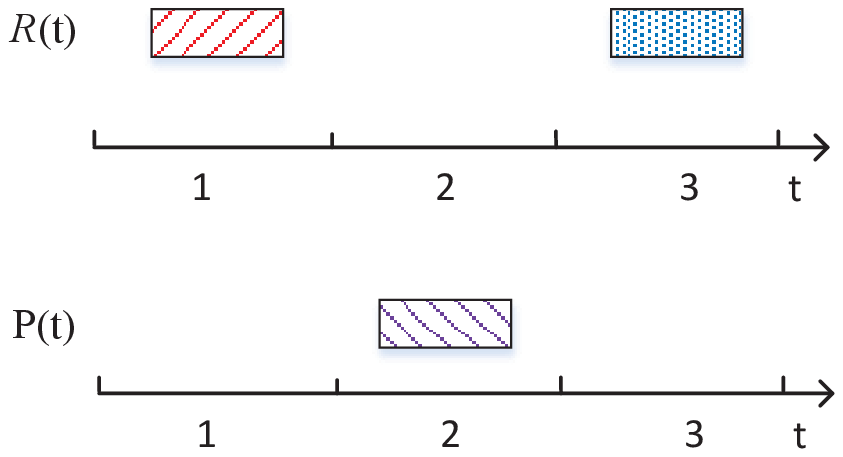}}}\quad\quad
\end{center}
 \caption{\small{Motivating Example. We consider $K=4$, $F=4$, $C=1$ and $\phi(x)=x^2$. Note that the blank square indicates that there is no file request.}
 }
\label{UserDemand}
\end{figure}
Problem $1$ is an infinite horizon average cost MDP.  According to Definition $4.2.2$ and Proposition $4.2.6$ in \cite{bertsekas}, we know that there exists an optimal policy that is unichain. Hence, in this paper, we restrict our attention to stationary unichain policies. Moreover, the MDP has finite state and action spaces as well as a bounded per-stage cost. Thus, there always exists a deterministic stationary unichain policy that is optimal and it is sufficient to focus on the deterministic stationary unichain policy space. In the following, we use $\mu$ to refer to a deterministic stationary unichain policy.

\section{Optimal Policy}
\subsection{Optimality Equation}
We can obtain the optimal joint pushing and caching policy $\mu^\ast$ through solving the following Bellman equation.
\begin{lemma}[Bellman equation]\label{Bellman equation} There exist a scalar $\theta$ and a value function\! $V(\cdot)$ satisfying
\begin{align}\label{average transmission cost Objective2}
&\theta\! + V(\textbf{X})\!=\min \limits_{(\mathbf{P}, \Delta \textbf{S})\in \mathbf{U}(\textbf{X})}  \Big\{ \phi\big(\sum_{f\in \mathcal{F}}(R_f + P_f)\big)\! +\sum_{\mathbf{A}' \in \bar{\mathcal{F}}^K} \prod_{k\in\mathcal{K}} q_{A_k,A_k'}^{(k)}V(\mathbf{A}',\textbf{S}\!+\Delta \textbf{S})\! \Big\},\nonumber\\&\hspace{120mm} \textbf{X} \!\in\! \bar{\mathcal{F}}^K \times \mathcal{S}^K,
\end{align}
where $R_f$ is given by (\ref{reactive}) and $\mathbf{A}'\! \triangleq\! (A_k')_{k\in\mathcal{K}}$. $\theta = \bar{\phi}^*$ is the optimal value of Problem $1$ for all initial system states $\textbf{X}(0)\! \in \!\bar{\mathcal{F}}^K\! \times\! \mathcal{S}^K$, and the optimal policy $\mu^*$ can be obtained from
\begin{align} \label{Optimal Policy1}
&\ \ \mu^*(\textbf{X})=\arg \min \limits_{(\mathbf{P}, \Delta \textbf{S})\in \mathbf{U}(\textbf{X})}  \Big\{ \phi\big(\sum_{f\in \mathcal{F}}(R_f + P_f)\big)+\sum_{\mathbf{A}' \in \bar{\mathcal{F}}^K} \prod_{k\in\mathcal{K}} q_{A_k,A_k'}^{(k)}V(\mathbf{A}',\textbf{S}+\Delta \textbf{S}) \Big\},\nonumber\\&\hspace{120mm} \textbf{X} \in \bar{\mathcal{F}}^K \times \mathcal{S}^K.
\end{align}
\end{lemma}
\begin{proof}
Please see Appendix A.
\end{proof}
From (\ref{Optimal Policy1}), we see that the optimal policy $\mu^*$ achieves a balance between the current transmission cost (i.e., the first term in the objective function of (\ref{Optimal Policy1})) and the future average transmission cost (i.e., the second term in the objective function of (\ref{Optimal Policy1})). Moreover, how $\mu^* = (\mu^*_P,\mu^*_{\Delta S})$ achieves the balance is illustrated in the following corollary.

\begin{corollary}
The optimal pushing policy $\mu^*_{P}$ is given by
\begin{align}\label{Transmission Policy}
&\mu_{P}^*(\textbf{X})=\arg \min \limits_{\mathbf{P}\in \mathbf{U}_p(\textbf{X})}\Big\{\phi\big(\sum_{f\in \mathcal{F}}(R_{f} + P_f)\big) + W(\textbf{X},\mathbf{P})\Big\}, \ \ \  \textbf{X} \in \bar{\mathcal{F}}^K \times \mathcal{S}^K,
\end{align}
 where $W(\textbf{X},\mathbf{P})\! \triangleq \!\min \limits_{\Delta \textbf{S} \in \mathbf{U}_{\Delta S}(\textbf{X}, \mathbf{P})} \!\sum_{\mathbf{A}'\in \bar{\mathcal{F}}^K}\!\prod_{k\in\mathcal{K}}q_{A_k,A_k'}^{(k)}\!V(\mathbf{A}',\textbf{S}+\Delta \textbf{S})$ is a nonincreasing function of $\mathbf{P}$.
 Furthermore, the optimal caching policy $\mu_{\Delta S}^*$ is given by
  \begin{align}\label{Cache Policy}
     &\mu_{\Delta S}^*(\textbf{X},\mu_{P}^*(\textbf{X}))= \arg \min \limits_{\Delta \textbf{S}\in\mathbf{U}_{\Delta S}(\textbf{X},\mu_{P}^*(\textbf{X}))}
     \sum_{\mathbf{A}'\in \bar{\mathcal{F}}^K}\prod_{k\in\mathcal{K}}q_{A_k,A_k'}^{(k)}V(\mathbf{A}',\textbf{S}+\Delta \textbf{S}),\ \textbf{X}\in \bar{\mathcal{F}}^K \times \mathcal{S}^K ,
  \end{align}
where $\mu_{P}^*$ is obtained from (\ref{Transmission Policy}).
\end{corollary}
\begin{proof}
(\ref{Transmission Policy}) and (\ref{Cache Policy}) follow directly from (\ref{Optimal Policy1}). In addition, if $\textbf{P}_1\preceq \textbf{P}_2$,\footnote{The notion $\preceq$ indicates the component-wise $\leq$.} $\mathbf{U}_{\Delta S}(\mathbf{X}, \textbf{P}_1)\subseteq \mathbf{U}_{\Delta S}(\mathbf{X}, \textbf{P}_2)$, implying that $W(\mathbf{X},\textbf{P}_1) \geq W(\mathbf{X},\textbf{P}_2)$. The proof ends.
\end{proof}

\begin{remark}[Balance between Current Transmission Cost and Future Average Transmission Cost]
Note that the current transmission cost $\phi\big(\sum_{f\in \mathcal{F}}(R_f + P_f)\big)$ increases with $\mathbf{P}$ and the future average transmission cost $W(\textbf{X},\mathbf{P})$ decreases with $\mathbf{P}$. Thus, the optimal pushing policy $\mu^*_{P}$ in (\ref{Transmission Policy}) achieves the perfect balance between the current transmission cost and the future average transmission cost for all $\textbf{X}$. In addition, from (\ref{Cache Policy}), we learn that the optimal caching policy $\mu_{\Delta S}^*$ achieves the lowest future average transmission cost under the optimal pushing policy~$\mu^*_{P}$.
\end{remark}

From Lemma~$1$ and Corollary~$1$, we note that $\mu^*$ depends on system state $\textbf{X}$ via the value function $V(\cdot)$. Obtaining $V(\cdot)$ involves solving the equivalent Bellman equation in (\ref{average transmission cost Objective2}) for all $\textbf{X}$, and there generally exist only numerical results which cannot offer many design insights \cite{bertsekas}. In addition, obtaining numerical solutions using value iteration or policy iteration is usually infeasible for practical implementation, due to the curse of dimensionality \cite{bertsekas}. Therefore, it is desirable to study optimality properties of $\mu^*$ and exploit these properties to design low-complexity policies with promising performance.

\subsection{Optimality Properties}
First, we analyze the impact of cache size $C$ on the optimal average transmission cost $\theta$. For ease of exposition, we rewrite $\theta$ as a function of cache size $C$, i.e., $\theta(C)$, and obtain the following lemma based on coupling and interchange arguments \cite{Jsac}.
\begin{lemma}[Impact\! of Cache\! Size]\label{the effect of cache size} $\theta(C)$ decreases with $C$ when $C\! <\! F$ and $\theta(C)\!=\!0$ when $C\! \geq\! F$.
\end{lemma}
\begin{proof}
Please see Appendix B.
\end{proof}
\begin{remark}[Tradeoff between Cache Size and Bandwidth Utilization]
As illustrated in Fig.~$2$, a lower average transmission cost always corresponds to a higher bandwidth utilization. Hence, Lemma $2$ reveals the tradeoff between the cache size and the bandwidth utilization.
\end{remark}

In the following, we focus on the case of $C<F$. By analyzing the partial monotonicity of value function $V(\cdot)$, we obtain the next lemma.
\begin{lemma}[Transient System States]\label{Transient} Any $\textbf{X} = (\mathbf{A},\mathbf{S})$ with $\mathbf{S}\notin \check{\mathcal{S}}^K$ is transient under $\mu^*$, where $\check{\mathcal{S}}\! \triangleq\!\big\{(S_{f})_{f\in\mathcal{F}}:\! \sum_{f=1}^F\! S_{f}\!=\!C\big\}$.
\end{lemma}

\begin{proof}
Please see Appendix C.
\end{proof}
\begin{remark}[Reduction of System State Space and Caching Action Space] Lemma $3$ reveals that the optimal policy $\mu^*$ makes full use of available storage resources. Also, considering the expected sum cost over the infinite horizon incurred by a transient state is finite and negligible in terms of average cost, we restrict our attention to the reduced system state space $\bar{\mathcal{F}}^K\times \check{\mathcal{S}}^K$ without loss of optimality.
Also, the cache update constraint in (\ref{Storage Evolution Constraint2}) is replaced with $\sum_{f\in \mathcal{F}}S_{k,f}+\Delta S_{k,f}= C$, and thus the caching action space can be further reduced.
\end{remark}

\begin{remark}[Computational Complexity and Implementation Requirement] To obtain the optimal policy $\mu^*$ from (\ref{Optimal Policy1}) under the reduced system state space given in Lemma~\ref{Transient}, we need to compute $V(\textbf{X})$, $\textbf{X}\! \in\! \bar{\mathcal{F}}^K\!\times\! \check{\mathcal{S}}^K$, by solving a system of $\big((F+1)\binom{F}{C}\big)^K$ equations in (\ref{average transmission cost Objective2}), the number of which increases exponentially with the number of users $K$ and combinatorially with the number of files $F$ as\! well\! as the cache size $C$. Moreover, given $V(\cdot)$, computing $\mu^*(\textbf{X})$ for all $\textbf{X}$ involves brute-force search over the action space $\mathbf{U}(\textbf{X})$, which requires complexity of $\mathcal {O}\big(K2^F\binom{F}{C}\big)$. In practice, $K$, $F$ and $C$ are relatively large, and hence the complexity of computing $\mu^*$ is not acceptable. Besides, the implementation of $\mu^*$ requires a centralized controller and system state information, resulting in large signaling overhead.
\end{remark}

\section{Low-Complexity Decentralized Policy}
To reduce the computational complexity and achieve decentralized implementation without much signaling overhead, we first approximate the value function $V(\cdot)$ in (\ref{average transmission cost Objective2}) by the sum of per-user per-file value functions. Based on the approximate value function, we obtain a low-complexity decentralized policy for practical implementation.






\subsection{Value Approximation}

To alleviate the curse of dimensionality in computing $V(\cdot)$, for all $\textbf{X}\!\in\! \bar{\mathcal{F}}^K \times \check{\mathcal{S}}^K$, motivated by \cite{Cui1,Cui2}, we approximate $V(\textbf{X})$ in (\ref{average transmission cost Objective2}) as follows:
\begin{equation}\label{two stage approximation1}
V(\textbf{X}) \approx \check{V}(\textbf{X}) = \sum_{k\in\mathcal{K}}\sum_{f_k\in \mathcal{F}:S_{k,f_k}=1} \check{V}^1_k(X_k^1),
\end{equation}
where $X_k^1 \triangleq (A_k,f_k) \in \bar{\mathcal{F}}\times \mathcal{F}$ and for all $k\in \mathcal{K}$, $\check{V}^1_k(X_k^1)$, $X_k^1 \in \bar{\mathcal{F}}\times \mathcal{F}$  satisfy:
\begin{align}\label{per user per content11}
&\theta_k^1 + \check{V}^1_k(X_k^1)= \phi'(X_k^1)+ \min \limits_{\Delta S_k^1 \in U_k^1(X_k^1)}\sum_{A_k'\in \bar{\mathcal{F}}} q^{(k)}_{A_k,A_k'}\check{V}^1_k(A_k',f_k'),\ X_k^1 \in \bar{\mathcal{F}}\times \mathcal{F}.
\end{align}
Here, $\phi'(X_k^1)\! \triangleq\! \frac{1}{K}\left(\frac{\phi(1)}{C}\!-\!\phi\big(\textbf{1}(A_k=f_k)\big)\right)\textbf{1}(A_k\! \neq\! 0)$,\footnote{\textbf{1}($\cdot$) represents the indicator function throughout this paper.} $U_k^1(X_k^1) \triangleq \{0,-\textbf{1}(A_k \notin \{0,f_k\})\}$ and $f_k' \triangleq (1+\!\Delta S_k^1)f_k\!-\!\Delta S_k^1 A_k$.
The equation in (\ref{per user per content11}) corresponds to the Bellman equation of a per-user per-file MDP for user $k$ with unit cache size. $\theta_k^1$ and $\check{V}^1_k(\cdot)$ denote the average cost and value function of the per-user per-file MDP for user $k$, respectively. Specifically, at time slot $t$, $X_k^1(t)=(A_k(t),f_k(t))$ denotes the system state, where $A_k(t) \in \bar{\mathcal{F}}$ denotes the demand state and $f_k(t) \in \mathcal{F}$ denotes the cached file; $\Delta S_k^1(t)\in U_k^1(X_k^1)$ denotes the caching action; the demand state $A_k(t)$ evolves according to the Markov chain $\{A_k(t)\}$ and the cache state $f_k(t)$ evolves according to $f_k(t+1) = (1+ \Delta S_k^1(t))f_k(t)\! -\! \Delta S_k^1(t) A_k(t)$; $\phi'(X_k^1(t))$ denotes the per-stage cost. The $K$ per-user per-file MDPs are obtained from the original MDP by eliminating the couplings among the $K$ users and the $C$ cache units of each user, which are due to the multicast transmission and the cache size constraint, respectively.


In the following, we characterize the performance of the value approximation in (\ref{two stage approximation1}) from the perspectives of the average transmission cost and the complexity reduction, respectively.
First, by analyzing the relaxation from the original MDP to the $K$ per-user per-file MDPs, we have the following relationship between the average cost of the original MDP and the sum of the average costs of the $K$ per-user per-file MDPs.
\begin{lemma}\label{lower bound achieve}
$\theta(C)$ and $\theta_k^1$, $k\in \mathcal{K}$ satisfy that $\theta(C)\geq C\sum_{k\in \mathcal{K}}\theta_k^1$.
\end{lemma}
\begin{proof}
Please see Appendix D.
\end{proof}
In addition, note that obtaining $V(\textbf{X}),\ \textbf{X}\!\in\! \bar{\mathcal{F}}^K\!\times \check{\mathcal{S}}^K$ requires to solve a system of $\big((F+1)\binom{F}{C}\big)^K$ equations given in (\ref{average transmission cost Objective2}), while obtaining $\check{V}^1_k(X_{k}^1)$, $X_k^1\in \bar{\mathcal{F}}\times\mathcal{F}$, $k \in \mathcal{K}$ only requires to solve a system of $KF(F+1)$ equations given in (\ref{per user per content11}). Therefore, under the value function approximation in (\ref{two stage approximation1}), the non-polynomial computational complexity is eliminated.
\begin{remark}The linear value function approximation adopted in (\ref{two stage approximation1}) differs from most existing approximation methods. Firstly, different from the traditional linear approximation in \cite{Linear}, our approach is not based on specific basis functions. Secondly, compared with the randomized approach proposed in \cite{Cui1,Cui2}, our approach leads to a lower bound of the optimal average cost as illustrated in Lemma $4$.
\end{remark}
\subsection{Low-complexity Decentralized Policy}
By replacing $V(\textbf{X})$ in (\ref{average transmission cost Objective2}) with $\check{V}(\textbf{X})$ in (\ref{two stage approximation1}), the minimization problem in (\ref{Optimal Policy1}) which determines the optimal policy $\mu^*$ is approximated by:
\begin{problem}[Approximate Joint Pushing and Caching  Optimization]
For all $\textbf{X} \in \bar{\mathcal{F}}^K\times \check{\mathcal{S}}^K$,
\begin{align}
&\min_{(\mathbf{P},\Delta \mathbf{S})} \! \ \ \ \ \ \ \ \ \ \ \varphi(\mathbf{P},\Delta \mathbf{S})\! \nonumber\\
&\ \ \ s.t. ~\ \ \ \ \ \ \ \ (\ref{reactive}),(\ref{once}),(\ref{twice}),(\ref{Cache Decision Constraint}),(\ref{Storage Evolution Constraint2}),(\ref{Storage Evolution Constraint1}),\nonumber
\end{align}
where $\varphi(\mathbf{P},\Delta \mathbf{S})\! \triangleq\!\phi\big(\sum_{f\in \mathcal{F}} (R_f\!+\! P_f)\big)\!+\!\sum_{k\in\mathcal{K}}\! \sum_{f\in \mathcal{S}_k'}\! g_k(A_k,f)$, $\mathcal{S}_k' \triangleq \{f\in\mathcal{F}:S_{k,f}\!+\! \Delta S_{k,f}=1\}$ and $g_k(A_k,f)\triangleq \sum_{A_k'\in \bar{\mathcal{F}}} q^{(k)}_{A_k,A_k'}\check{V}^1_k(A_k',f)$. Let $\check{\mu}^*(\textbf{X})$ denote the corresponding optimal solution.
\end{problem}



Note that due to the coupling among $K$ users incurred by the multicast transmission, solving Problem~$2$ still calls for complexity of $\mathcal {O}\big(K2^F\binom{F}{C}\big)$ and centralized implementation with system state information, which motivates us to develop a low-complexity decentralized policy. Specifically, given system state $\textbf{X}$, first ignore the multicast opportunities in pushing and separately optimize the per-user pushing action of each user $k$ under given state $X_k$ and reactive transmission $\textbf{R}$. Then, the server gathers the information of the per-user pushing actions of all the users and multicasts the corresponding files. Next, each user optimizes its caching action given the files obtained from the multicast transmissions. The details are mathematically illustrated as follows.

First, for all $k\in \mathcal{K}$, replace $P_f$ with $P_{k,f}$ and by adding constraints $P_f = P_{k,f}$, we obtain an equivalent problem of Problem~$2$. The constraint in (\ref{once}) is rewritten as
\begin{equation}\label{per user push11}
P_{k,f}(t) \leq 1 - R_{f}(t),\ \ \ f\in \mathcal{F},\ k \in \mathcal{K},
\end{equation}
which is to guarantee that each file $f\in \mathcal{F}$ is transmitted at most once to user $k$ at each time slot $t$.
The constraints in (\ref{twice}) and (\ref{Cache Decision Constraint}) can be replaced by
\begin{equation}\label{per user push22}
P_{k,f}(t) \leq 1 - S_{k,f}(t),\ \ \ f\in \mathcal{F},\ k \in \mathcal{K},
\end{equation}
\begin{equation}\label{Cache Decision Constraint11}
-S_{k,f}(t) \leq \Delta S_{k,f}(t) \leq R_{f}(t)+ P_{k,f}(t), \ f \in \mathcal{F},\ k \in \mathcal{K}.
\end{equation}
Via omitting the constraints $P_f = P_{k,f}$, $k\!\in\! \mathcal{K}$, we attain a relaxed optimization problem of Problem~$2$. Given $\textbf{R}$, by (\ref{per user push11}) (\ref{per user push22}) and (\ref{Cache Decision Constraint11}), the relaxed problem can be decomposed into $K$ separate subproblems, one for each user, as shown in Problem~$3$. 


\begin{problem}[Pushing Optimization for User $k$]
For all state $X_k$ and $\textbf{R}$,
\begin{align}
&\varphi_k^* \triangleq \min_{\mathbf{P}_k}\  \Big \{\frac{\phi\big(\!\sum_{f\in \mathcal{F}} (R_f\!+\! P_{k,f})\!\big)}{K}\!+\!W_k(X_k,\mathbf{R}+\mathbf{P}_k) \!\Big\}\nonumber\\
&\ \ \ \ \ \ \ \ s.t. ~\ \ \ \ \ \ \ \ (\ref{reactive}),(\ref{Storage Evolution Constraint2}),(\ref{Storage Evolution Constraint1}),(\ref{per user push11}),(\ref{per user push22}),(\ref{Cache Decision Constraint11}),\nonumber
\end{align}
where $W_k(X_k,\mathbf{R}\!+\!\mathbf{P}_k)\!$ $\triangleq$\!\!~$\min_{\Delta \mathbf{S}_k \in U_{\Delta S,k}(X_k,\mathbf{R}+\mathbf{P}_k)}$ $\sum_{f \in \mathcal{S}_k'}\!g_k(A_k,f)$. Let $\textbf{P}^*_k$ denote the optimal solution.
\end{problem}

Then, we obtain $\textbf{P}_k^*$ as follows. Denote with $\textbf{y}_k(p_k) \triangleq (y_{k,f}(p_k))_{f\in\mathcal{F}}$ the optimal pushing action for user $k$ when the number of pushed files for user $k$ is $p_k$. From the definition of $W_k(X_k,\mathbf{R}\!+\!\mathbf{P}_k)$, we learn that user $k$ always pushes the first $p_k$ files with the minimum values of $g_k(A_k,f)$, $f\!\in\! \{f\!\in\! \mathcal{F}: S_{k,f}\!+\!R_f\!=\!0\}$. Hence, we obtain $\textbf{y}_k(p_k)$ as follows.
Given  $X_k$ and $\textbf{R}$ in (\ref{reactive}), sort the elements in $\mathcal G_k(X_k,\textbf{R})\triangleq \{g_k(A_k,f):  S_{k,f}\!+R_f\! = 0,  f \!\in\! \mathcal{F}\}$ in ascending order, let $f_{k,i}$ denote the index of the file with the $i$-th minimum in $\mathcal G_k(X_k,\textbf{R})$, and we have
\begin{align}\label{update}
 y_{k,f}(p_k) =
 &\left\{
 \begin{array}{l}
  1, \ \  f=f_{k,i}, \ i \leq p_k, \\
  0, \ \ \ \text{otherwise},
 \end{array}
 \right. f\in \mathcal F, \  p_k \in \{0,1,\cdots, |\mathcal G_k(X_k,\textbf{R})|\}.
\end{align}
Based on (\ref{update}), we can easily obtain $\textbf{P}_k^*$, as summarized below.

\textbf{Optimal Solution to Problem~3:} \textit{For all state $X_k$ and $\textbf{R}$, $\textbf{P}_k^* = (y_{k,f}(p_k^*))_{f\in\mathcal{F}}$, where $y_{k,f}(p_k^*)$ is given by (\ref{update}) and $p_k^*$ is given by}
 \begin{align}\label{pno}
 &p_k^* \triangleq \arg\ \! \min_{p_k} \big\{ \phi\big(\sum_{f\in \mathcal{F}} R_f + p_k\big)+ W_k(X_k,\mathbf{R}\!+\!\mathbf{y}_k(p_k))\big\}. 
 \end{align}


Next, based on $\textbf{P}_k^*$, $k\!\in\!\mathcal{K}$, we propose a low-complexity decentralized policy, denoted as $\check{\mu}\! \triangleq\! (\check{\mu}_P,\check{\mu}_{\Delta S})$, which reconsiders the multicast opportunities in pushing. Specifically, for all $\textbf{X}\! \in\! \bar{\mathcal{F}}^K\!\times\! \check{\mathcal{S}}^K$, we have $\check{\mu}_P(\textbf{X}) \triangleq (\check{P}_{f})_{f\in\mathcal{F}}$ and $\check{\mu}_{\Delta S}(\textbf{X})\! \triangleq\!  (\Delta \check{\textbf{S}}_k)_{k\in\mathcal{K}}$, where
\begin{equation}\label{decp}
 \check{P}_{f} \triangleq \max_{k\in \mathcal{K}}\ P^*_{k,f}, \ f\in\mathcal{F},
\end{equation}
\begin{align}\label{decs}
\Delta \check{\textbf{S}}_k \triangleq  \arg \  \min_{\Delta \mathbf{S}_k\in \check{\mathbf{U}}_{\Delta S}(X_k, \mathbf{R}+\check{\mu}_P(\textbf{X}))}\sum_{f \in \mathcal{S}_k'} g_k(A_k,f),\ k\in \mathcal{K}.
\end{align}

Finally,\! we characterize\! the performance of $\check{\mu}$. Lemma\!~$5$ illustrates the relationship among the optimal values of Problem~$2$ and Problem~$3$ as\! well\! as the objective\! value\! of Problem~$2$ at\! $\check{\mu}$.

\begin{lemma} For all $\textbf{X} \in \bar{\mathcal{F}}^K\times \check{\mathcal{S}}^K$, $\sum_{k\in \mathcal{K}}\varphi_k^*\!\leq\! \varphi(\check{\mu}^*(\textbf{X}))\! \leq \! \varphi(\check{\mu}(\textbf{X}))$, where the equality holds if and only if $X_{k_1}\!=\!X_{k_2}$ and $\textbf{Q}_{k_1}\!=\!\textbf{Q}_{k_2}$ for all $k_1\!\in\! \mathcal{K}$ and $k_2\!\in\! \mathcal{K}$.
\end{lemma}
\begin{proof}
Due to the relaxation from Problem $2$ to Problem $3$, the action space becomes larger and thus we can show the first inequality directly. Due to the suboptimality of $\check{\mu}$, the second inequality holds. Furthermore, when $X_{k_1}\!=\!X_{k_2}$ and $\textbf{Q}_{k_1}\!=\!\textbf{Q}_{k_2}$ for all $k_1\!\in\! \mathcal{K}$ and $k_2\!\in\! \mathcal{K}$, $\textbf{P}_{k_1}^* = \textbf{P}_{k_2}^* = \check{\mu}_P(\textbf{X})$ and thus the equality holds. We complete the proof.
\end{proof}

\begin{remark}[Computational Complexity and Implementation Requirement]
Given $\check{V}^1_k(\cdot)$, for all $\textbf{X} \in \bar{\mathcal{F}}^K\times \check{\mathcal{S}}^K$, the complexity of computing $\check{\mu}(\textbf{X})$ is $\mathcal {O}\big(\!KF\!\log(F)\!\big)$ much lower than that of computing $\mu^*(\textbf{X})$, i.e., $\mathcal {O}\big(K2^F\binom{F}{C}\big)$. Furthermore, we note that $\check{\mu}$ can be implemented in a decentralized manner. Specifically, first, each user submits its request $A_k$ if $A_k\in \{f\in \mathcal{F}:S_{k,f}=0\}$. 
Then the server broadcasts the corresponding file indexes $\{f\in \mathcal{F}: \max_{k\in \mathcal{K}: A_k = f} (1-S_{k,f})=1\}$, which implies $\mathbf{R}$. Next, based on $X_k$ and $\mathbf{R}$, user $k$ computes $\mathbf{P}_k^*$ and reports it to the server. Finally, the server obtains $\check{\mu}(\textbf{X})$ and transmits the files in $\{f\in \mathcal{F}:R_f + \check{P}_f \geq 1\}$, based on which user $k$ obtains $\Delta \check{\mathbf{S}}_k$.
\end{remark}
\section{Online Decentralized Algorithm}
To implement the low-complexity decentralized policy $\check{\mu}$ proposed in Section V, we need to compute $g_k(X_k^1)=\sum_{A_k'\in \bar{\mathcal{F}}} q^{(k)}_{A_k,A_k'}\check{V}^1_k(A_k',f_k)$, requiring priori knowledge of the transition matrices of the $K$ user demand processes, i.e., $\textbf{Q}_k$, $k\!\in\!\mathcal{K}$. In this section, we propose an online decentralized algorithm (ODA) to implement $\check{\mu}$ via Q-learning \cite{bertsekas}, when $\textbf{Q}_k$ is unknown.
\begin{algorithm}
 \caption{Online Decentralized Algorithm (ODA)}
\label{Algorithm2}
\small{\begin{algorithmic}[1]
\STATE \textbf{Initialization}. Set $t=0$. Each user $k$ initializes $\mathcal{Q}_{k,t}(\cdot)$.
\STATE  \textbf{Per-User Per-File Q-factor Update}.
 At the beginning of the $t$th slot, $t\geq 1$, each user $k$ updates $\mathcal{Q}_{k}(\cdot)$ according to
\begin{align}\label{q-learning}
&\mathcal{Q}_{k,t}(X_k^1,\Delta S_k^1)\!=\! \mathcal{Q}_{k,t-1}(X_k^1,\Delta S_k^1)\! +\!\gamma(v_{k,t-1}(A_k))\textbf{1}(A_k(t\!-\!1)=A_k)\big( \phi'(X_k^1) + \nonumber\\ & \min \limits_{\Delta S_{k}^{1'}\! \in U_k^1(X_{k}^{1'})}\! \mathcal{Q}_{k,t-1}\big(\!X^{1'}_{k},\Delta\! S_{k}^{1'}\!)-\!\mathcal{Q}_{k,t-1}(X_k^1,\Delta S_k^1)\!-\!\min \limits_{\Delta S_{k,0}^1 \in U_k^1(X_{k,0}^1)}
\mathcal{Q}_{k,t-1}(X_{k,0}^1,\Delta S_{k,0}^1)\big), \nonumber \\ &\hspace{80mm}X_k^1 \in \bar{\mathcal{F}} \times \mathcal{F}, \Delta S_k^1 \in U_k^1(X_{k}^1),
\end{align}
where $X_k^1\triangleq (A_k,f_k)$, $X_{k}^{1'} \triangleq (A_k(t),f_k')$ and $v_{k,t}(A_k)$ denotes the number of times that $A_k \in \bar{\mathcal{F}}$ has been requested by user $k$ up to $t$, and then updates $g_{k,t}(X_k^1)$, $X_k^1\! \in\! \bar{\mathcal{F}}\times\mathcal{F}$ according to
 \begin{align}\label{gupdate}
&g_{k,t}(X_k^1) = \min_{\Delta S_{k,0}^1 \in U_k^1(X_{k,0}^1)} \mathcal{Q}_{k,t}(X_{k,0}^1,\Delta S_{k,0}^1)+ \mathcal{Q}_{k,t}(X_k^1,0) - \phi'(X_k^1),\ X_k^1\! \in\! \bar{\mathcal{F}}\times\mathcal{F}.
\end{align}
\STATE \textbf{Reactive Transmission Message}. Each user $k$ submits $A_k(t)$ if $A_k(t)\!\in\! \{f\in \mathcal{F}:\! S_{k,f}(t)\!=\!0\}$. Then the server broadcasts the file indexes $\{f\in \mathcal{F}: \max_{k\in \mathcal{K}: A_k(t) = f} (1-S_{k,f}(t))=1\}$.
\STATE  \textbf{Per-User Pushing Computation}. Each user $k$ constructs $\mathbf{R}(t)$.
Given $X_k(t)$, $\mathbf{R}(t)$ and $g_{k,t}(A_k(t),f_k)$, $f_k \in \mathcal{F}$, user $k$ computes $\textbf{P}^*_k(t)$ and then reports it to the server.



\STATE  \textbf{Multicast Transmission at Server}. The server obtains $\check{P}_f(t)$ in (\ref{decp}) and multicasts the files in $\{f\in \mathcal{F}:R_f(t) + \check{P}_f(t) = 1\}$.
\STATE  \textbf{Per-User Caching}. Each user $k$ updates its own cache state $\textbf{S}_k(t)$ according to $\Delta \check{\mathbf{S}}_k(t)$ in (\ref{decs}).
\STATE Set $t\leftarrow t+1$ and go back to Step $2$.
\end{algorithmic}}
\end{algorithm}

First, introduce the\! Q-factor $\mathcal{Q}_k(\!X_k^1,\!\Delta \!S_k^1)$ of the per-user\! per-file state-action pair $(\!X_k^1,\!\Delta S_k^1)$ as
 \begin{align}\label{Qfactor}
&\theta_k^1 + \mathcal{Q}_k(X_k^1,\Delta S_k^1)\triangleq \phi'(X_k^1)\!+\!\sum_{A_k'\in \bar{\mathcal{F}}} q^{(k)}_{A_k,A_k'}V_k^1(A_k',f_k'),\ \!X_k^1\! \in\! \bar{\mathcal{F}}\!\times\! \mathcal{F},\ \Delta S_k^1 \in U_k^1(X_k^1).
\end{align}
By (\ref{per user per content11}) and (\ref{Qfactor}), we have
\begin{align}\label{valueupdate1}
&\check{V}^1_k(X_k^1) = \min_{\Delta S_k^1 \in U_k^1(X_k^1)} \mathcal{Q}_k(X_k^1,\Delta S_k^1), \ X_k^1 \in \bar{\mathcal{F}}\times \mathcal{F},
\end{align}
 \begin{align}\label{Qfactor2}
&\theta_k^1 + \mathcal{Q}_k(X_k^1,\Delta S_k^1)= \phi'(X_k^1) +\!\sum_{A_k'\in \bar{\mathcal{F}}} q^{(k)}_{A_k,A_k'}\min_{\Delta S_{k}^{1'} \in U_k^1(X_k^{1'})} \mathcal{Q}_k((A_k',f_k'),\Delta S_{k}^{1'}),\nonumber\\ & \hspace{95mm} X_k^1 \in \bar{\mathcal{F}}\times \mathcal{F}, \Delta S_k^1 \in U_k^1(X_k^1).
\end{align}
Then, by (\ref{Qfactor}) and (\ref{Qfactor2}), we can express $g_k(X_k^1)$ as a function of the Q-factor $\mathcal{Q}_k(\cdot)$, $k\in \mathcal{K}$, i.e.,
\begin{align}\label{g}
&g_k(X_k^1) = \theta_k^1+ \mathcal{Q}_k(X_k^1,0) - \phi'(X_k^1), \ \ \ X_k^1\in\! \bar{\mathcal{F}}\times \mathcal{F}.
\end{align}
Recall that $\check{\mu}$ given in (\ref{decp}) and (\ref{decs}) is expressed in terms of $g_k(X_k^1)$, $k\in \mathcal{K}$. From (\ref{g}), we learn that $\check{\mu}$ can be determined by the Q-factor $\mathcal{Q}_k(\cdot)$, $k\in \mathcal{K}$. Considering that $\check{\mu}$ cannot be obtained directly via minimizing the corresponding Q-factor, the standard Q-learning algorithm cannot be used to implement $\check{\mu}$ online. Next, we propose the ODA, as shown in Algorithm~$1$, to learn $\mathcal{Q}_k(\cdot)$ and implement $\check{\mu}$ online when $\textbf{Q}_k$, ${k\in\mathcal{K}}$ are unknown. In particular, the stepsize $\gamma(\cdot)$ in the ODA satisfies that $0\!<\!\gamma(n)<\infty$, $\sum_{n=1}^\infty\! \gamma(n)\! =\! \infty$ and $\sum_{n=1}^\infty \gamma^2(n)\! <\! \infty$. Based on the convergence results of Q-learning in \cite{bertsekas}, we can easily show that $\lim_{n\rightarrow \infty}\mathcal{Q}_{k,n}(X_k^1,\Delta S_k^1)\!=\!\mathcal{Q}_{k}(X_k^1,\Delta\! S_k^1)$ almost surely. 

 \begin{remark}[Illustration of the ODA] The proposed ODA differs from the standard Q-learning algorithm in the following two facets. Firstly, for each per-user per-file MDP, at each time slot, instead of updating the Q-factor at the currently sampled state-action pair, it updates the Q-factors at a set of state-action pairs with the current demand state, thereby speeding up the convergence. 
Secondly, when learning the Q-factors of the $K$ per-user per-file MDPs, it implements a policy which cannot be directly obtained from the optimal policies of the $K$ per-user per-file MDPs.
\end{remark}

\section{Numerical Results}
In this section, we first evaluate the convergence of our proposed ODA and then compare it with five baselines. Specifically, we consider three baselines which operate based on priori knowledge of $\textbf{Q}_k$, $k\in\mathcal{K}$: the aforementioned MP caching policy in Section III, local most popular (LMP) caching policy in which at each time slot $t$, each user $k$ stores the $C$ files in $\{f\in \mathcal{F}: R_f(t)+S_{k,f}(t) =1\}$ with the largest transition probabilities given current demand state $A_k(t)$ \cite{WebCorrelation}, as well as joint threshold-based pushing \cite{Thres} and local most popular caching policy (TLMP) where at each time slot $t$, the server pushes the file $f^*(t) \triangleq \arg \max_{f\in \mathcal{F}}\! \sum_{k\in \mathcal{K}:S_{k,f}(t)=0}\! q^{(k)}_{A_k(t),f}$ if and only if $\sum_{f\in \mathcal{F}}\! R_f(t)$ is below a threshold $T$, and each user implements the LMP caching policy. Note that MP and LMP are of the same complexity order, i.e., $\mathcal {O}\big(\!KF\!\log(F)\!\big)$, while the complexity order of TLMP is $\mathcal {O}\big(\!KF^2\!\log(F)\!\big)$.
In addition, we consider two other baselines which operate without priori knowledge of $\textbf{Q}_k$, $k\in\mathcal{K}$, and make caching decisions based on instantaneous user demand information, i.e., LRU and LFU. They are of the same complexity order, i.e., $\mathcal {O}\big(\!KF\big)$. In the simulation, we consider $\textbf{Q}_k = \textbf{Q}$ for all $k\in \mathcal{K}$ and adopt $\textbf{Q} \triangleq (q_{i,j})_{i\in \bar{\mathcal{F}},j\in \bar{\mathcal{F}}}$ similar to the demand model in \cite{WebCorrelation}, where $q_{i,j}$ is given by
\begin{align}\label{matrix}
 q_{i,j} \triangleq
 &\left\{
  \begin{array}{l}
  Q_0, \ \ \  i \in \bar{\mathcal{F}},\ j = 0,\\
  (1-Q_0)\frac{\frac{1}{j^{\gamma}}}{\sum_{j'=1}^F\frac{1}{j'^{\gamma}} }, \ \ i=0,\  j \in \mathcal{F},\\
  (1-Q_0)\frac{1}{N}, \ \ i \in \mathcal{F},\  j = (i+q)\!\mod(F+1) ,\ q \in \{1,2,\cdots,N\}, \\
  0, \ \ \ \text{otherwise}.
 \end{array}
 \right.
\end{align}
Note that $\textbf{Q}$ is parameterized by $\{Q_0,\gamma,N\}$. Specifically, $Q_0$ denotes the transition probability of requesting nothing given any current file request. The transition probability of requesting any file $f\in \mathcal{F}$ given no current file request, i.e., $i=0$, is modeled by a Zipf distribution parameterized by $\gamma$. For any $i \in \mathcal{F}$, we assign a set of neighboring files, i.e., $\mathcal{N}_i \triangleq \{f\in \mathcal{F}:f = (i+q)\!\mod(F+1) ,\ q = 1,2,\cdots,N\}$, where $N$ represents the number of neighbors. Then, the transition probability of requesting any file $f \in \mathcal{N}_i$ given the current file request $i\in \mathcal{F}$ is modeled by the uniform distribution. The transition probability of requesting any file $f \notin \mathcal{N}_i$ given current file request $i \in \mathcal{F}$ is zero. In the simulation, we set $F=100$, $T\!=\!1$, $N\!=\!2$, $Q_0\! =\! 0.2$ and $\gamma\!=\!0.5$ unless otherwise stated.

\begin{figure*}
\begin{center}
  \subfigure[Convergence at $K\! = \!10$,\! $C\!=\!~5$.]
 {\resizebox{7.1cm}{!}{\includegraphics{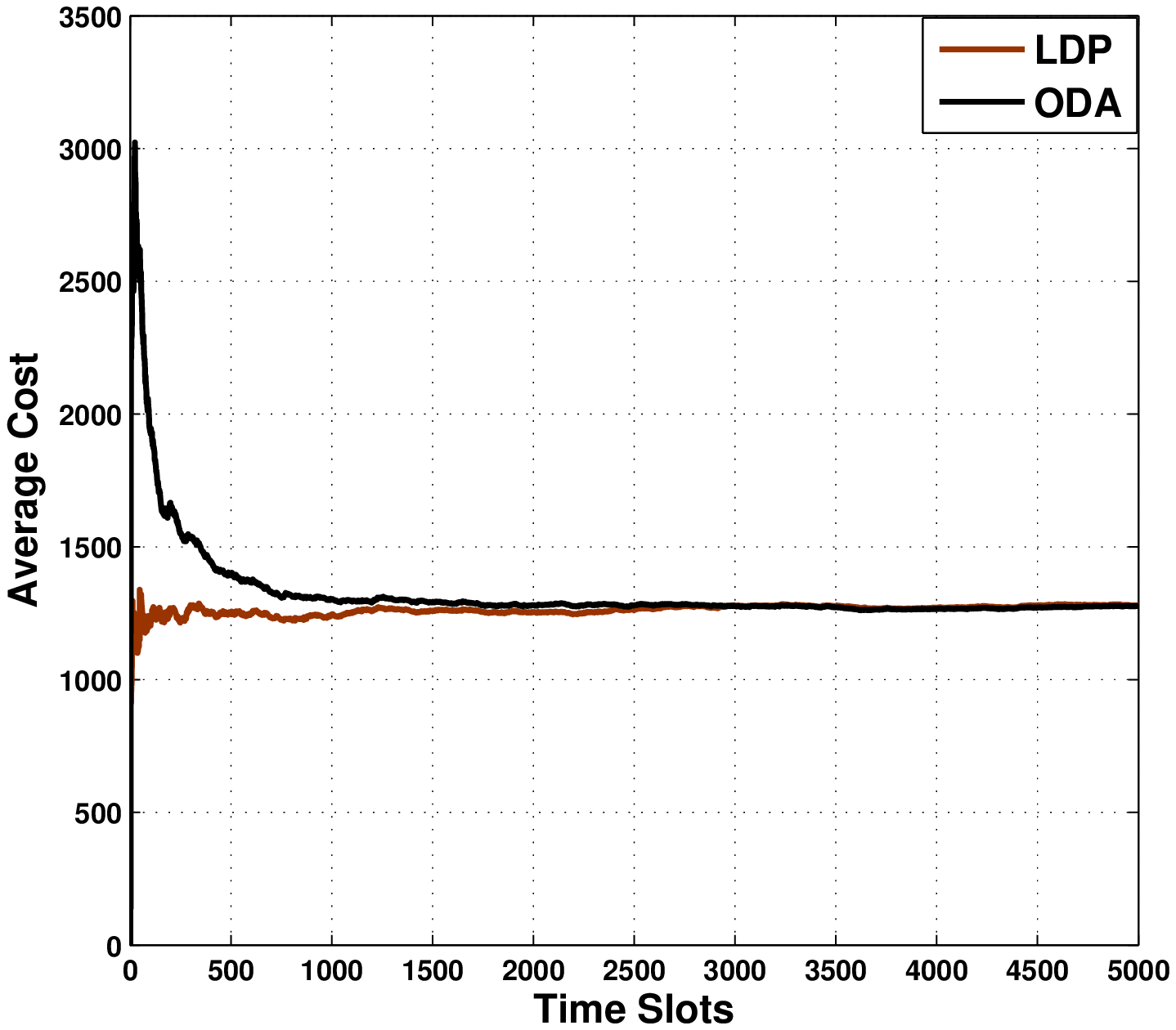}}}\quad\quad
 \subfigure[Cache size at $K = 10$.]
  {\resizebox{7.6cm}{!}{\includegraphics{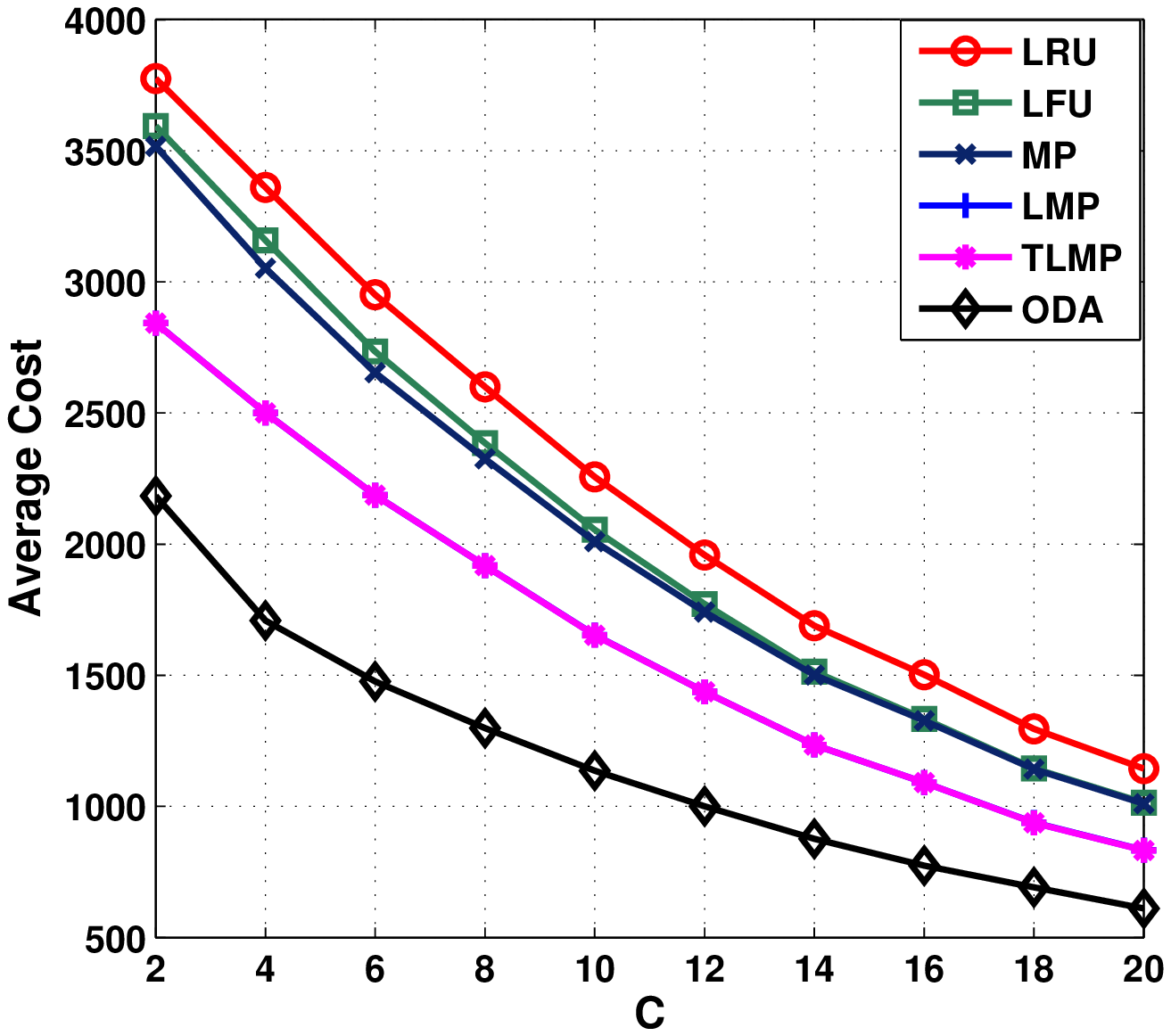}}}\quad\quad
 \subfigure[Number of users at $C = 10$.]
 {\resizebox{7cm}{!}{\includegraphics{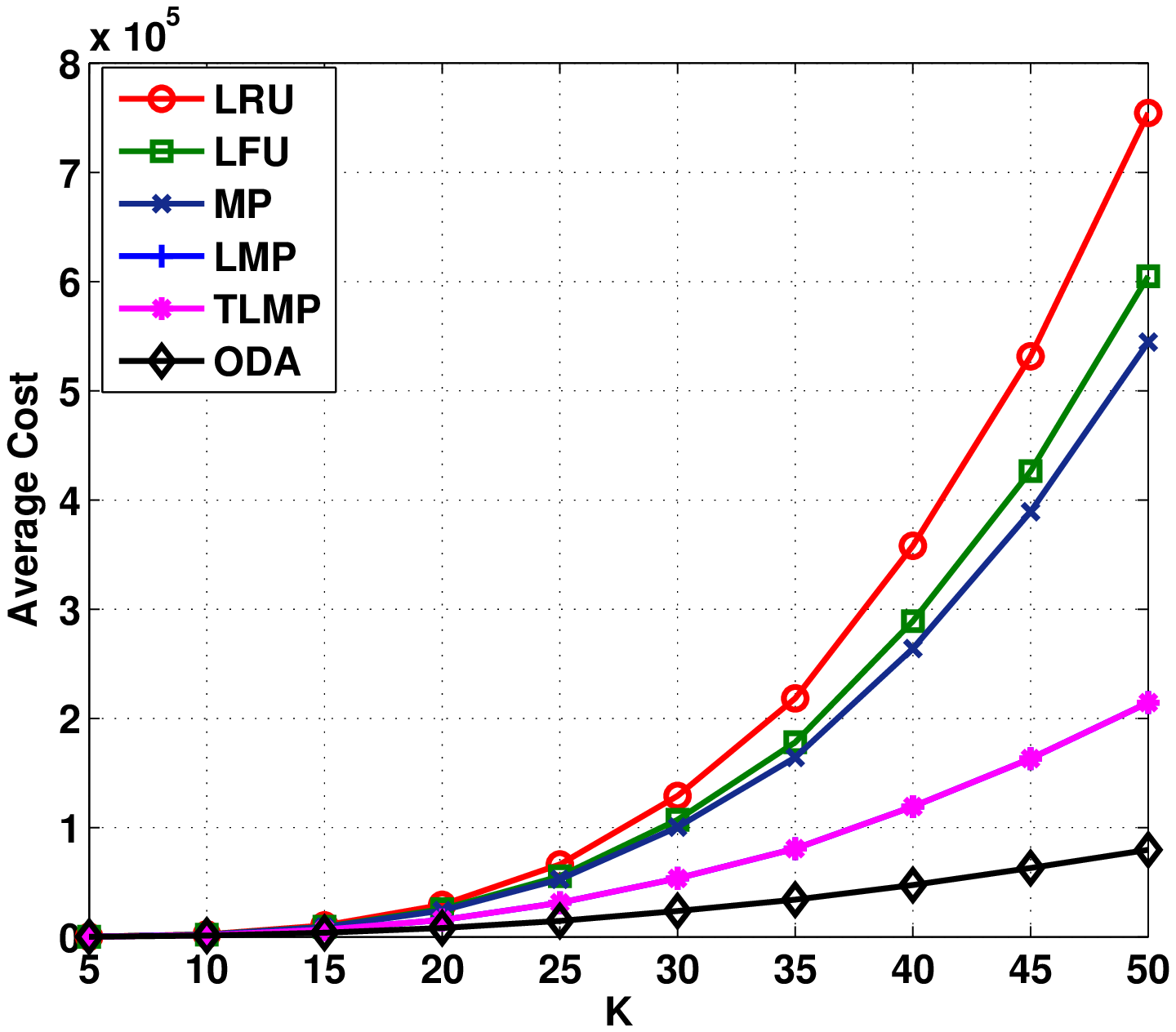}}}\quad\quad
 \subfigure[N at $C\!=~\!10$,~$K\!=\!10$,~$\gamma=1$.]
 {\resizebox{7cm}{!}{\includegraphics{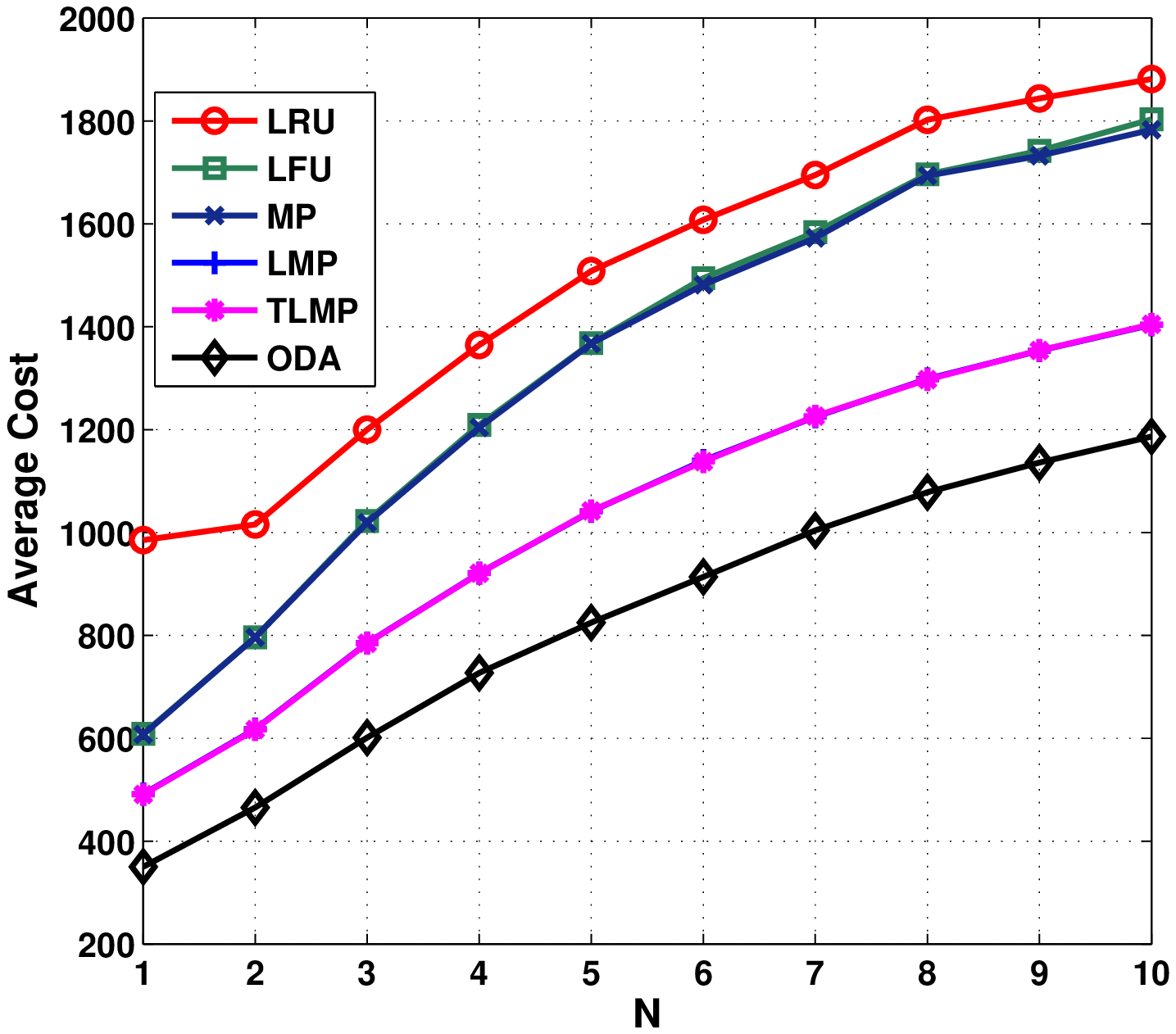}}}\quad\quad
  \subfigure[$\gamma$ at $C\!=~\!10$,~$K\!=\!10$.]
  {\resizebox{7.8cm}{!}{\includegraphics{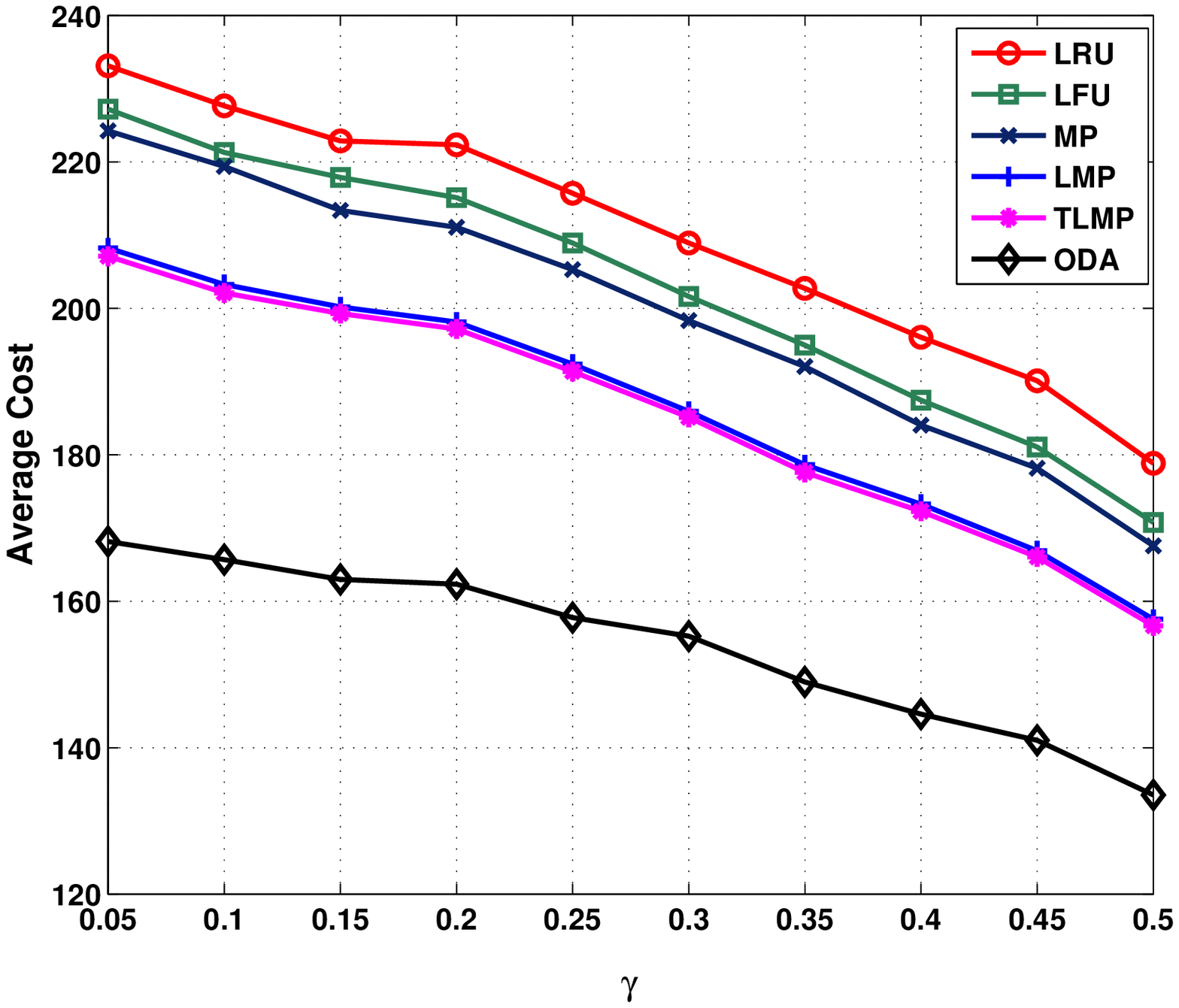}}}\quad\quad
   \subfigure[$Q_0$ at $C\!=~\!10$,~$K\!=\!10,N=15$.]
 {\resizebox{7cm}{!}{\includegraphics{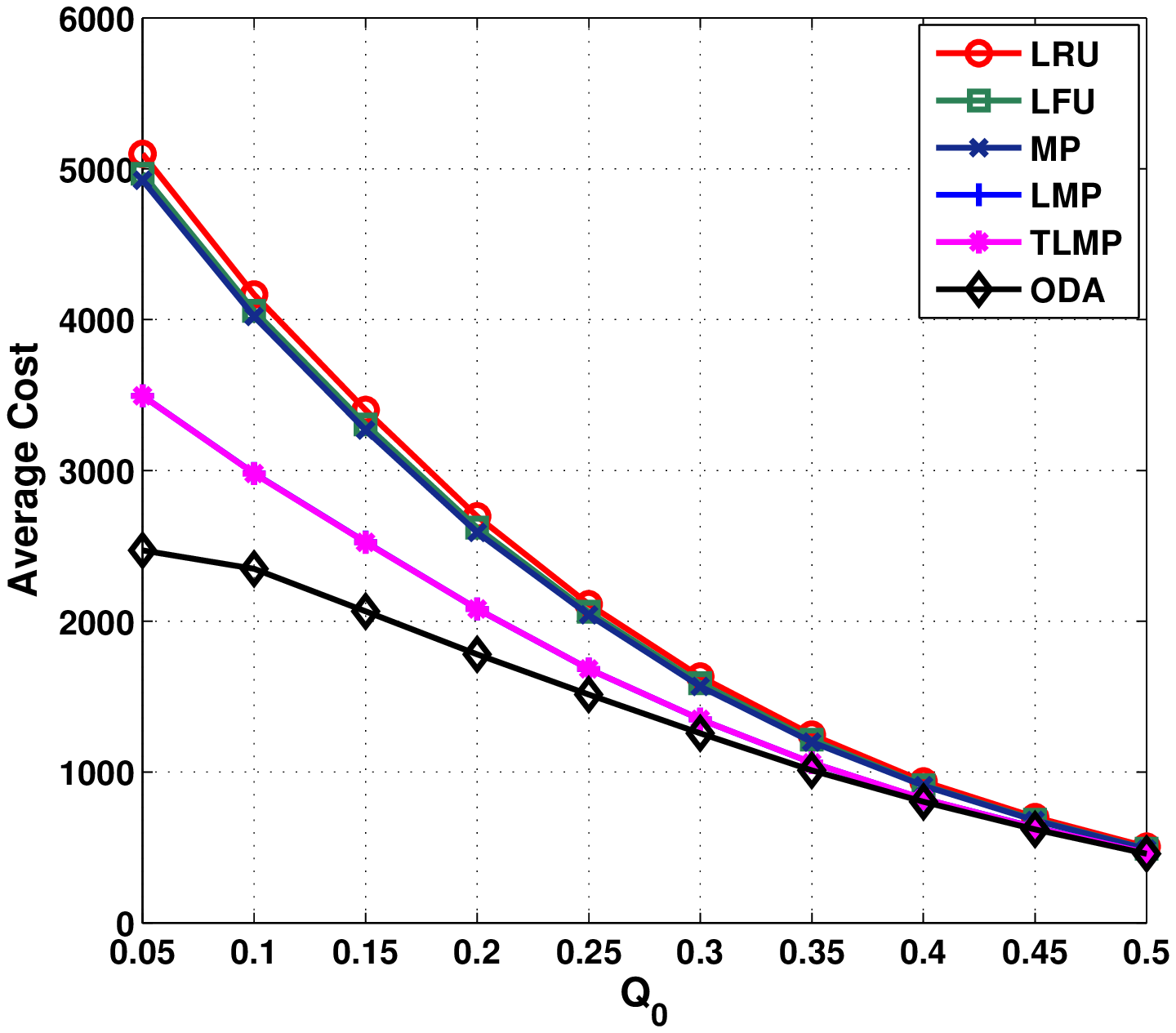}}}\quad\quad
 \end{center}
   \caption{\small{Convergence and average cost versus cache size $C$ and number of users $K$}.}
\label{Effect2}
\end{figure*}

Fig.~$4$~(a) shows that the proposed ODA converges quite fast. Fig.~$4$~(b)-(f) illustrate the average cost versus several system parameters. We observe that LMP behaves better than MP, LRU and LFU, mainly due to the fact that LMP considers the temporal correlation of each user demand process. The performance of TLMP is almost the same as that of LMP as its pushing and caching policies are not intelligently designed. The last not the least, our proposed ODA significantly outperforms the five baselines, primarily due to the fact that ODA takes into account both the asynchronous feature and temporal correlation of file requests and jointly designs both pushing and caching. Additionally, ODA achieves a good balance between the current transmission cost and the future average transmission cost.

Specifically, Fig.~$4$~(b) illustrates the average cost versus the cache size. Intuitively, the average cost monotonically decreases with the cache size. We can also see that our proposed ODA achieves good performance gains over the five baselines even at a small cache size. 
Fig.~$4$~(c) illustrates the average cost versus the number of users $K$. As expected, the average cost monotonically increases with the number of users, since the traffic load increases with the number of users. Furthermore, we can see that the performance gains of our proposed ODA over the five baselines increase with the number of users. Therefore, ODA behaves much robuster against the change of the number of users than the five baselines. 
Fig.~$4$~(d)-(g) illustrate the average cost versus the parameters of the transition matrix of the user demand process, i.e., $N$, $\gamma$ and $Q_0$. Specifically, Fig.~$4$~(d) illustrates the average cost versus the number of neighbors $N$. We can see that the average cost monotonically increases with the number of neighbors. This is because the user demand processes become less predictable as $N$ becomes larger. 
Fig.~$4$~(e) illustrates the average cost versus the Zipf exponent $\gamma$. We see that the average cost monotonically decreases with $\gamma$. This is because as $\gamma$ grows, the probability that a requested file is popular and is cached becomes larger.  Fig.~$4$~(f) illustrates the average cost versus $Q_0$. The average cost decreases with $Q_0$, mainly due to the fact that the traffic load becomes lighter. As $Q_0$ decreases, the performance gains of ODA over the five baselines become larger, which again indicates that ODA is much robuster against the change of the traffic load than the five baselines.

\section{Conclusions}

 In this paper,  we formulate the bandwidth utilization maximization problem via joint pushing and caching as an infinite horizon average cost MDP. By structural analysis, we show how the optimal policy balances the current transmission cost with the future average transmission cost. In addition, we show that  the optimal policy achieves a  tradeoff between the cache size and the bandwidth utilization. By a linear approximation of the value function and relaxation techniques, we develop a decentralized policy with polynomial complexity. Moreover, we propose an online decentralized algorithm to implement the proposed low-complexity decentralized policy when priori knowledge of user demand processes is unknown. Finally, using numerical results, we demonstrate the advantage of the proposed solutions over some existing designs. 
\section*{Appendix A: Proof of Lemma~\ref{Bellman equation}}\label{proof of Bellman equation}
First, we show that the Weak Accessibility (WA) condition holds for our system. Consider any two system states $\mathbf{X}^1 \triangleq (\textbf{A}^1,\mathbf{S}^1) \in \bar{\mathcal{F}}^K \times \mathcal{S}^K$ and $\mathbf{X}^2 \triangleq (\textbf{A}^2,\mathbf{S}^2) \in \bar{\mathcal{F}}^K \times \mathcal{S}^K$. Recall that for any $k\in \mathcal{K}$, $\{A_k(t)\}$ is an irreducible Markov chain. Thus, there exists an integer $t'\geq 1$ such that $\Pr[A_k(t')=A_k^2|A_k(0) = A_k^1] \geq 0$. In addition, there exists a policy $\bar{\mu} \triangleq (\bar{\mu}_P,\bar{\mu}_{\Delta S})$ such that $\bar{\mu}_{\Delta S}(\mathbf{A},\mathbf{S}^1) = \mathbf{S}^2$ and $\bar{\mu}_{\Delta S}(\mathbf{A},\mathbf{S}^2) = \mathbf{S}^2$ for all $\mathbf{A}\in \bar{\mathcal{F}}$. Hence, $\Pr[\textbf{X}(t')=\textbf{X}^2|\textbf{X}(0)=\textbf{X}^1, \bar{\mu}] = \Pr[\textbf{A}(t')=\textbf{A}^2|\textbf{A}(0)=\textbf{A}^1] \geq 0$, i.e., $\mathbf{X}^2$ is accessible from $\mathbf{X}^1$ under policy $\bar{\mu}$. By Definition $4.2.2$ in \cite{bertsekas}, we conclude that WA holds for the MDP.
Thus, by Proposition $4.2.3$ and Proposition $4.2.1$ in \cite{bertsekas}, the optimal average costs of the MDP in Problem~$1$ for all initial system states are the same and the solution $(\theta,V(\cdot))$ to the following Bellman equation exists:
\begin{align}\label{Average Cost Objective}
&\theta + V(\mathbf{X}) = \min \limits_{(\textbf{P}, \Delta \textbf{S})\in \mathbf{U}(\mathbf{X})}\Big\{\phi(\sum_{f\in \mathcal{F}}(R_f+P_f))+ \sum_{\mathbf{X}' \in \bar{\mathcal{F}}^K \times \mathcal{S}^K} \Pr [\mathbf{X}'|\mathbf{X},(\textbf{P}, \Delta \textbf{S})] V(\mathbf{X}')\Big\},\forall\  \mathbf{X},
\end{align}
where $\mathbf{X}' \triangleq (\textbf{A}', \mathbf{S}')$ and $R_f$ is given by (\ref{reactive}).
Furthermore, the optimal policy $\mu^*$ is given by
\begin{align} \label{Optimal Policy2}
\mu^*(\mathbf{X})&=\arg \min \limits_{(\textbf{P}, \Delta \textbf{S})\in \mathbf{U}(\mathbf{X})}  \Big\{ \phi\big(\sum_{f\in \mathcal{F}}(R_f + P_f)\big)\!+\!\sum_{\mathbf{X}' \in \bar{\mathcal{F}}^K \times \mathcal{S}^K} \Pr [\mathbf{X}'|\mathbf{X},(\textbf{P}, \Delta \textbf{S})]V(\mathbf{X}') \Big\}, \forall\  \mathbf{X}.
\end{align}
Note that the transition probability of the system state is given by:
\begin{align}\label{System State Transition}
&\Pr[\mathbf{X}'|\mathbf{X},(\textbf{P}, \Delta \textbf{S})] = \Pr [(\textbf{A}',\textbf{S}')|(\textbf{A},\textbf{S}),(\textbf{P}, \Delta \textbf{S})] = \Pr[\textbf{A}'|\textbf{A}]\Pr[\textbf{S}'|\textbf{S},\Delta \textbf{S}]\nonumber \\
& \ \ \ \ \  \ \ \ \ \ \ \ \ \ \ \ \ \ \ \ \ \ =  \begin{cases}
 \prod_{k\in \mathcal{K}} q^{(k)}_{A_k,A_k'}, &  \textbf{S}'= \textbf{S}+\Delta \mathbf{S},\\
 0, & \textbf{S}'\neq \textbf{S}+\Delta \mathbf{S}.\\
 \end{cases}
\end{align}
By substituting (\ref{System State Transition}) into (\ref{Average Cost Objective}) and (\ref{Optimal Policy2}), we obtain the Bellman equation in (\ref{average transmission cost Objective2}) and the optimal policy in (\ref{Optimal Policy1}), respectively. Therefore, we complete the proof.
\section*{Appendix B: Proof of Lemma~\ref{the effect of cache size}}\label{proof of impact of cache size}

First, for any $C_1$ and $C_2$ such that $C_1 < C_2 < F$, we show $\theta(C_1) > \theta(C_2)$ based on the coupling and interchange arguments \cite{Jsac}. Consider two independent MDP systems, i.e., System~$1$ and System~$2$, which have the same transition matrix of user demand processes, i.e., $(\textbf{Q}_{k})_{k\in\mathcal{K}}$, and numbers of files and users, i.e., $F$ and $K$, but have different cache sizes, denoted as $C_1$ and $C_2$, where $C_1 < C_2 <F$. Suppose $\textbf{A}^1(t) = \textbf{A}^2(t)$ for all time slot $t$. That is, the two systems are under \emph{the same sample paths} of the user demand processes. In addition, both systems adopt the same pushing action at each time slot $t$, denoted as $\textbf{P}^{1*}(t)$, which is the optimal pushing action for System~$1$ and a feasible pushing action for System~$2$ (due to $C_1 < C_2$). On the other hand, the two systems may have different caching actions at each time slot $t$. Consider any $\textbf{S}^1(0)\in \mathcal{S}_1 \triangleq \{(S_{k,f})_{k\in \mathcal{K},f\in \mathcal{F}}: \sum_{f\in \mathcal{F}} S_{k,f} = C_1\} $ and $\textbf{S}^2(0)\in \mathcal{S}_2 \triangleq \{(S_{k,f})_{k\in \mathcal{K},f\in \mathcal{F}}: \sum_{f\in \mathcal{F}} S_{k,f} = C_2\} $ such that $\textbf{S}^1(0) \preceq \textbf{S}^2(0)$. The cache state of System $1$ evolves according to $\textbf{S}^1(t+1) = \textbf{S}^1(t)+\Delta \textbf{S}^{1*}(t)$, where $\Delta \textbf{S}^{1*}(t)$ denotes the optimal caching action for System~$1$ at each time slot $t$. System~$2$ implements a caching policy such that at each time slot $t$, $\textbf{S}^2(t)\in \mathcal{S}_2$ and $\textbf{S}^1(t) \preceq \textbf{S}^2(t)$. This holds because that $C_1 < C_2$, $\textbf{S}^1(0) \preceq \textbf{S}^2(0)$ and $\textbf{P}^{2}(t)=\textbf{P}^{1*}(t)$. Based on the facts that $\textbf{A}^1(t) = \textbf{A}^2(t)$ and $\textbf{S}^1(t) \preceq \textbf{S}^2(t)$, by (\ref{reactive}) we have $\textbf{R}^2(t) \preceq \textbf{R}^1(t)$, i.e., $R_f^2 \leq R_f^1$, $f\in \mathcal{F}$, implying $\phi \!\left(\!\sum_{f\in \mathcal{F}} R_f^2(t)\!+\!P^{1*}_f(t) \!\right)\! \leq \!\phi \!\left(\!\sum_{f\in \mathcal{F}} R_f^1(t)\!+\!P^{1*}_f(t)\!\right)\!$ at each time slot $t$.
Considering that $C_1<C_2<F$ and $\textbf{S}^1(t) \preceq \textbf{S}^2(t)$, for each user $k$, there exists at least a file $f_k \in \mathcal{F}$ such that $S_{k,f_k}^1(t)=0 < S_{k,f_k}^2(t)=1$. For each $k\in \mathcal{K}$, since $\{A_k\}$ is irreducible, $f_k$ can be requested by user $k$ within a finite average number of transitions. Therefore, there exists at least a time slot $t$ such that $A_{k}(t) =f_k$ for all $k\in \mathcal{K}$. By (\ref{reactive}), $R_{f_k}^1(t) > R_{f_k}^2(t)$ holds for all $k\in \mathcal{K}$ and thus $\phi \!\left(\!\sum_{f\in \mathcal{F}} R_f^2(t)\!+\!P^{1*}_f(t) \!\right)\! < \!\phi \!\left(\!\sum_{f\in \mathcal{F}} R_f^1(t)\!+\!P^{1*}_f(t)\!\right)\!$. Thus, $\theta(C_1) > \theta'(C_2)$, where $\theta'(C_2)$ denotes the average cost for System~$2$ under the aforementioned policy for System~$2$. Hence, $\theta(C_1)>\theta(C_2)$. 
 Secondly, when $C \geq F$, intuitively, at each time slot, $\sum_{f\in \mathcal{F}} R_f = 0$ can be satisfied. Hence, $\theta(C) = 0$. The proof ends.



\section*{Appendix C: Proof of Lemma~\ref{Transient}}
We prove Lemma \ref{Transient} based on the partial monotonicity of value function $V(\cdot)$ w.r.t. the system cache state $\mathbf{S}$ shown using relative value iteration algorithm (RVIA) and mathematical induction.

First, we introduce RVIA \cite{bertsekas}. For all $\mathbf{X} \in \bar{\mathcal{F}}^K \times \mathcal{S}^K$, let $V_n(\mathbf{X})$ denote the value function in the $n$th iteration, where $n = 0,1,\cdots$. Define
\begin{align}\label{state action function}
&J_{n+1}(\mathbf{X},u_n)\triangleq \phi\left(\sum_{f=1}^{F} R_{n,f}+P_{n,f}\right)+\sum_{\textbf{A}'\in \bar{\mathcal{F}}^K}\prod_{k\in \mathcal{K}} q^{(k)}_{A_k,A_k'}V_n(\textbf{A}',\textbf{S}+ \Delta \textbf{S}_n),
\end{align}
where $u_n \triangleq (\textbf{P}_n,\Delta \textbf{S}_n)$ denotes the system action under state $\textbf{X}$ in the $n$th iteration. Note that $J_{n+1}(\mathbf{X},u_n)$ corresponds to the R.H.S of the Bellman equation in (\ref{average transmission cost Objective2}). We refer to $J_{n+1}(\mathbf{X},u_n)$ as the state-action cost function in the $n$th iteration. Under RVIA, $V_n(\mathbf{X})$ evolves according to
\begin{align}\label{iteration equation}
V_{n+1}(\mathbf{X}) &= \min \limits_{u_n} J_{n+1}(\mathbf{X},u_n) - \min \limits_{u_n} J_{n+1}(\mathbf{X}^\S,u_n)\nonumber\\
&= J_{n+1}(\mathbf{X},\mu_n^*(\mathbf{X})) - J_{n+1}(\mathbf{X}^\S,\mu_n^*(\mathbf{X}^\S)),\ \  \mathbf{X} \in \bar{\mathcal{F}}^K \times \mathcal{S}^K
\end{align}
where $J_{n+1}(\mathbf{X},u_n)$ is given by (\ref{state action function}), $\mu_n^*$ denotes the optimal policy that attains the minimum of the first term in (\ref{iteration equation}) in the $n$th iteration and $\mathbf{X}^\S \in \bar{\mathcal{F}}^K \times \mathcal{S}^K$ is some fixed state.
By Proposition $4.3.2$ in \cite{bertsekas}, for all $\mathbf{X} \in \bar{\mathcal{F}}^K \times \mathcal{S}^K$, the generated sequence $\{V_{n}(\mathbf{X})\}$ converges to $V(\mathbf{X})$ given in the Bellman equation in (\ref{average transmission cost Objective2}) under any initialization of $V_{0}(\mathbf{X})$, i.e.,
\begin{equation}
\lim_{n\rightarrow \infty} V_{n}(\mathbf{X}) = V(\mathbf{X}),\ \mathbf{X} \in \bar{\mathcal{F}}^K \times \mathcal{S}^K,
\end{equation}
where $V(\mathbf{X})$ satisfies the Bellman equation in (\ref{average transmission cost Objective2}). 



Next, we prove the partial nonincreasing monotonicity of $V(\cdot)$ w.r.t. the system cache state $\mathbf{S}$, i.e., for all $\textbf{S}^1, \textbf{S}^2 \in \mathcal{S}^K$ such that $\textbf{S}^1 \preceq \textbf{S}^2$, $V(\textbf{A},\textbf{S}^1) \geq V(\textbf{A},\textbf{S}^2)$ for all $\textbf{A} \in \bar{\mathcal{F}}^K$. Based on RVIA, it is equivalent to show that for all $\textbf{S}^1, \textbf{S}^2 \in \mathcal{S}^K$ such that $\textbf{S}^1 \preceq \textbf{S}^2$,
\begin{equation}\label{monotonicity of value function}
V_{n}(\textbf{A},\textbf{S}^1) \geq V_{n}(\textbf{A},\textbf{S}^2),
\end{equation}
holds for all $n=0,1,\cdots$. We now prove (\ref{monotonicity of value function}) based on mathematical induction. First, we initialize $V_0(\mathbf{X})=0$ for all $\mathbf{X} \in \bar{\mathcal{F}}^K \times \mathcal{S}^K$. Thus, we have $V_{0}(\textbf{A},\textbf{S}^1) \geq V_{0}(\textbf{A},\textbf{S}^2)$, meaning (\ref{monotonicity of value function}) holds for $n=0$. Then, assume (\ref{monotonicity of value function}) holds for some $n\geq0$. Denote with $(\textbf{P}_n^1, \Delta \textbf{S}_n^1)$ the optimal action under $(\mathbf{A},\mathbf{S}^1)$, i.e., $\mu_n^*(\mathbf{A},\mathbf{S}^1)=(\textbf{P}_n^1, \Delta \textbf{S}_n^1)$, and denote with $(\textbf{P}_n^2, \Delta \textbf{S}_n^2)$ the optimal action under $(\mathbf{A},\mathbf{S}^2)$, i.e., $\mu_n^*(\mathbf{A},\mathbf{S}^2) = (\textbf{P}_n^2, \Delta \textbf{S}_n^2)$. Define $\Delta\textbf{S}_n' \triangleq (\Delta S_{n,k,f}')_{k\in \mathcal{K},f\in\mathcal{F}}$ where
\begin{align}\label{cache update in n iteration}
& \Delta S_{n,k,f}'
  \triangleq  \begin{cases}
 \Delta S^1_{n,k,f}, &  S^2_{k,f} + \Delta S^1_{n,k,f} \leq 1,\\
 0, & S^2_{k,f} + \Delta S^1_{n,k,f} > 1,\\
 \end{cases}\ k\in \mathcal{K},f\in\mathcal{F}.
\end{align}
From (\ref{cache update in n iteration}), $(\textbf{P}_n^1,\Delta\textbf{S}_n')$ is a feasible action under $(\mathbf{A},\mathbf{S}^2)$. From (\ref{iteration equation}), we have
\begin{align}\label{n+1 iteration}
&V_{n+1}(\textbf{A},\textbf{S}^2)= J_{n+1}((\textbf{A},\textbf{S}^2),\mu_n^*(\mathbf{A},\mathbf{S}^2))-\min \limits_{u_n} J_{n+1}(\mathbf{X}^\S,u_n)\nonumber\\
&\ \ \ \ \ \ \  \ \ \ \ \ \ \overset{(a)}{\leq}J_{n+1}((\textbf{A},\textbf{S}^2),(\textbf{P}_n^1,\Delta\textbf{S}_n'))-\min \limits_{u_n} J_{n+1}(\mathbf{X}^\S,u_n)\nonumber\\
&\ \ \ \ \ \ \  \ \ \ \ \ \
= \phi(\sum_{f\in \mathcal{F}} R_{n,f}^2+P_{n,f}^1)+\sum_{\textbf{A}'\in \bar{\mathcal{F}}^K}\prod_{k\in \mathcal{K}} q^{(k)}_{A_k,A_k'}V_n(\textbf{A}',{\textbf{S}^2}+\Delta\textbf{S}_n') -\min \limits_{u_n} J_{n+1}(\mathbf{X}^\S,u_n)\nonumber\\
&\ \ \ \ \ \ \  \ \ \ \ \ \
\overset{(b)}{\leq}\phi(\sum_{f\in \mathcal{F}} R_{n,f}^1+P_{n,f}^1)+\sum_{\textbf{A}'\in \bar{\mathcal{F}}^K}\prod_{k\in \mathcal{K}} q^{(k)}_{A_k,A_k'}V_n(\textbf{A}',{\textbf{S}^2}+\Delta\textbf{S}_n') -\min \limits_{u_n} J_{n+1}(\mathbf{X}^\S,u_n)\nonumber\\
&\ \ \ \ \ \ \ \ \ \ \ \ \
\overset{(c)}{\leq}\phi(\sum_{f\in \mathcal{F}} R_{n,f}^1+P_{n,f}^1)+\sum_{\textbf{A}'\in \bar{\mathcal{F}}^K}\prod_{k\in \mathcal{K}} q^{(k)}_{A_k,A_k'}V_n(\textbf{A}',{\textbf{S}^1}+\Delta \textbf{S}_n^1) - \min \limits_{u_n} J_{n+1}(\mathbf{X}^\S,u_n) \nonumber\\
&\ \ \ \ \ \ \ \ \ \ \ \ \
\overset{(d)}{=}V_{n+1}(\textbf{A},\textbf{S}^1),
\end{align}
where (a) follows from the optimality of $\mu_n^*(\mathbf{A},\mathbf{S}^2)$ under $(\mathbf{A},\mathbf{S}^2)$ and the feasibility of $(\mathbf{P}_n^1,\Delta\textbf{S}_n')$ under $(\mathbf{A},\mathbf{S}^2)$ in the $n$th iteration. (b) follows from the fact that $\textbf{\textsl{R}}^2 \preceq \textbf{\textsl{R}}^1$ according to (\ref{reactive}) since $\textbf{S}^1 \preceq \textbf{S}^2$. (c) follows from $\textbf{S}^1+\Delta\textbf{S}_n^{1} \preceq \textbf{S}^2+\Delta\textbf{S}_n'$ (due to $\textbf{S}^1 \preceq \textbf{S}^2$ and (\ref{cache update in n iteration})) and (\ref{monotonicity of value function}). (d) follows from (\ref{state action function}) and (\ref{iteration equation}). Hence, (\ref{monotonicity of value function}) holds in the $(n+1)$th iteration. By induction, we show that $V_n(\textbf{A},\textbf{S}^1) \geq V_n(\textbf{A},\textbf{S}^2)$ holds for all $n=0,1,\cdots$. Thus, by RVIA, we conclude that $V(\textbf{A},\textbf{S}^1)\! \geq V(\textbf{A},\textbf{S}^2)$. Similarly, we can show that for all $\textbf{S}^1, \textbf{S}^2\! \in \mathcal{S}^K$ such that $\textbf{S}^1\! \preceq \textbf{S}^2$, if there exists at least a pair of $k$ and $f$ satisfying that $S^1_{k,f}\! < S^2_{k,f}$, $V(\textbf{A},\textbf{S}^1)\! > V(\textbf{A},\textbf{S}^2)$ for all $\textbf{A} \in \bar{\mathcal{F}}^K$.

Finally, based on the partial nonincreasing monotonicity of the value function, we prove Lemma~\ref{Transient}. For all state $(\mathbf{A},\mathbf{S}) \in \bar{\mathcal{F}}^K\times \check{\mathcal{S}}^K$, denote with $G(\mathbf{A}, \textbf{S}+\Delta \textbf{S}) \triangleq \sum_{\mathbf{A}'\in \bar{\mathcal{F}}^K}\prod_{k\in\mathcal{K}}q_{A_k,A_k'}^{(k)}\\V(\mathbf{A}',\textbf{S}+\Delta \textbf{S})$ the objective function in (\ref{Cache Policy}) for all $\Delta \textbf{S}\in\mathbf{U}_{\Delta S}(\textbf{X},\mu_{P}^*(\textbf{X}))$ and $\Delta\mathbf{S}^* = \mu^*_{\Delta S}(\mathbf{A},\textbf{S})$ the optimal caching action given in (\ref{Cache Policy}). Since for all $\textbf{S}^1, \textbf{S}^2 \in \mathcal{S}^K$ such that $\textbf{S}^1 \preceq \textbf{S}^2$, $V(\textbf{A},\textbf{S}^1) \geq V(\textbf{A},\textbf{S}^2)$ for all $\mathbf{A}\in \bar{\mathcal{F}}^K$, we have $G(\mathbf{A}, \textbf{S}+\Delta \textbf{S}^1) \geq G(\mathbf{A}, \textbf{S}+\Delta \textbf{S}^2)$ for all $\Delta \textbf{S}^1$, $\Delta \textbf{S}^2$ such that $\Delta \textbf{S}^1 \preceq \Delta \textbf{S}^2$. Furthermore, if there exists at least a pair of $k$ and $f$ satisfying that $S^1_{k,f} < S^2_{k,f}$, $G(\textbf{A},\textbf{S}+\Delta \textbf{S}^1) > G(\textbf{A},\textbf{S}+\Delta \textbf{S}^2)$.
In the sequel, we consider two cases:
\begin{itemize}
  \item Case i: $\mathbf{S} \notin \check{\mathcal{S}}^K$. First, we show that $\Delta\mathbf{S}^* \succeq \mathbf{0}$ by contradiction. Suppose that there exist $k'$ and $f'$ such that $\Delta S_{k',f'}^* < 0$. Denote $\Delta \mathbf{S} \triangleq (\Delta S_{k,f})_{k\in \mathcal{K}, f\in \mathcal{F}}$ where
      \begin{align}\label{special case1}
      & \Delta S_{k,f}
            \triangleq  \begin{cases}
       0, &  k = k', f= f',\\
           \Delta S^*_{k,f}, & \text{otherwise}.\\
       \end{cases}
      \end{align}
      Then, $\Delta \mathbf{S} \succeq \Delta \mathbf{S}^*$ with $\Delta S_{k',f'} > \Delta S^*_{k',f'}$. Hence, $G(\mathbf{A}, \textbf{S}+\Delta \textbf{S}) < G(\mathbf{A}, \textbf{S}+\Delta \textbf{S}^*)$ which contradicts the optimality of $\Delta \mathbf{S}^*$. Thus, we have $\Delta\mathbf{S}^* \succeq \mathbf{0}$. Furthermore, since $\mathbf{S}\notin \check{\mathcal{S}}^K$ and $C<F$, there always exists a demand state $\mathbf{A}$ such that $\sum_{f\in \mathcal{F}} R_f = \sum_{f\in \mathcal{F}} \max_{k\in \mathcal{K}: A_k=f}\  \big(1 - S_{k,f}\big)> 0$. By contradiction, we can show that for a state $(\mathbf{A},\mathbf{S})$ such that $\mathbf{S} \notin \check{\mathcal{S}}^K$ and $\sum_{f\in \mathcal{F}} R_f > 0$, there exists at least a pair of $k \in \mathcal{K}$ and $f \in \mathcal{F}$ such that $S_{k,f} = 0$ and $\Delta S_{k,f}^* = 1$.
     Recall that for all $k\in \mathcal{K}$, $\{A_k\}$ is irreducible, i.e., any demand state $A_k\in \bar{\mathcal{F}}$ can be visited within a finite average number of transitions. Thus, under the optimal policy $\mu^*$, if the system state starts from any state $(\mathbf{A},\mathbf{S})$ where $\mathbf{S}\notin \check{\mathcal{S}}^K$, the cache state $\mathbf{S}$ will transit into the set $\check{\mathcal{S}}^K$ and never move back.
  \item Case ii: $\mathbf{S} \in \check{\mathcal{S}}^K$. First, we show that $\mathbf{S}+\Delta \mathbf{S}^* \in \check{\mathcal{S}}^K$ by contradiction. Suppose that $\mathbf{S}+\Delta \mathbf{S}^* \notin \check{\mathcal{S}}^K$, then there exists $\Delta \mathbf{S}$ such that $\mathbf{S}+\Delta \mathbf{S}^* \preceq \mathbf{S}+\Delta \mathbf{S} \in \check{\mathcal{S}}^K$, implying that there exists at least a pair of $k$ and $f$ satisfying $S^1_{k,f}=0$ and $S^2_{k,f}=1$. Hence, $G(\textbf{A},\mathbf{S}+\Delta \mathbf{S}^*) < G(\textbf{A},\mathbf{S}+\Delta \mathbf{S})$ which contradicts the optimality of $\Delta \mathbf{S}^*$. Thus, under the optimal policy $\mu^*$, if started at all state $(\textbf{A},\mathbf{S})$ where $\mathbf{S}\in \check{\mathcal{S}}^K$, the system state shall never come to a state $ (\textbf{A},\mathbf{S})$ where $\mathbf{S}\notin \check{\mathcal{S}}^K$.
\end{itemize}
The proof ends.

\section*{Appendix D: Proof of Lemma~$4$}
We prove Lemma~$4$ by illustrating the relationship between the $K$ per-user per-file MDPs and the original MDP. First, we relax the action space $\mathbf{U}(\textbf{X})$ via ignoring the multicast opportunities (i.e., considering unicast transmissions) in both reactive transmission and pushing. Specifically, 
let $\mathbf{R}_k \triangleq (R_{k,f})_{f\in \mathcal{F}}\in \{0,1\}^F$ denote the reactive transmission action of user $k$, where
\begin{equation}\label{reactive2}
R_{k,f} \triangleq \textbf{1}\{A_k=f\}(1-S_{k,f}),\ f\! \in\! \mathcal{F},\ k \!\in\! \mathcal{K}.
\end{equation}
Denote with $\mathbf{P}_k \triangleq (P_{k,f})_{f\in \mathcal{F}}\in \{0,1\}^F$ the pushing action of user $k$. 
The per-user pushing action constraints are as follows:
\begin{equation}\label{per user push1}
P_{k,f}(t) \leq 1 - R_{k,f}(t),\ \ \ f\in \mathcal{F},\ k \in \mathcal{K},
\end{equation}
\begin{equation}\label{per user push2}
P_{k,f}(t) \leq 1 - S_{k,f}(t),\ \ \ f\in \mathcal{F},\ k \in \mathcal{K},
\end{equation}
where (\ref{per user push1}) is to guarantee that each file $f\in \mathcal{F}$ is not transmitted more than once to user $k$ at each time slot $t$ and (\ref{per user push2}) is to guarantee that a file $f$ which has already been cached in the storage of user $k$ is not pushed again at each time slot $t$. By omitting the coupling among users incurred by the multicast transmission, i.e., $R_f = \max_{k\in \mathcal{K}} R_{k,f}$ and $P_f = \max_{k\in \mathcal{K}} P_{k,f}$, we rewrite the cache update constraint in (\ref{Cache Decision Constraint}) as:
\begin{equation}\label{Cache Decision Constraint1}
-S_{k,f}(t) \leq \Delta S_{k,f}(t) \leq R_{k,f}(t)+ P_{k,f}(t), \ f \in \mathcal{F},\ k \in \mathcal{K}.
\end{equation}
In this way, we construct action space $\check{\mathbf{U}}(\textbf{X}) \triangleq \{(\mathbf{P}, \Delta \textbf{S}): \Delta \textbf{S} \in \check{\mathbf{U}}_{\Delta S}(\textbf{X},\mathbf{P}),\ \mathbf{P}\in \check{\mathbf{U}}_P(\textbf{X})\}$. Specifically, $\check{\mathbf{U}}_{\Delta S}(\textbf{X},\mathbf{P}) \triangleq \prod_{k\in \mathcal{K}} \check{\mathbf{U}}_{\Delta S,k}(X_k,\mathbf{R}_k+\mathbf{P}_k)$ where $\check{\mathbf{U}}_{\Delta S,k}(X_k,\mathbf{R}_k+\mathbf{P}_k) \triangleq \{(\Delta S_{k,f})_{f\in \mathcal{F}}: (\ref{Storage Evolution Constraint2})(\ref{Storage Evolution Constraint1})(\ref{Cache Decision Constraint1})\}$ denotes the caching action space of user $k$ and $\check{\mathbf{U}}_P(\textbf{X}) \triangleq \prod_{k\in \mathcal{K}} \check{\mathbf{U}}_{P}(X_k)$ where $\check{\mathbf{U}}_{P}(X_k) \triangleq \{(P_{k,f})_{f\in \mathcal{F}}: (\ref{per user push1})(\ref{per user push2})\}$ denotes the pushing action space of user $k$. Thus, $\check{\mathbf{U}}(\textbf{X}) = \prod_{k\in \mathcal{K}} \check{\mathbf{U}}_k(X_k)$, where $\check{\mathbf{U}}_k(X_k) \triangleq \{(\textbf{P}_k, \Delta \textbf{S}_k):\Delta \textbf{S}_k \in \check{\mathbf{U}}_{\Delta S,k}(X_k,\mathbf{R}_k+\mathbf{P}_k), \textbf{P}_k \in \check{\mathbf{U}}_{P,k}(X_k) \}$.

Then, we establish an MDP under the unicast transmission, named as unicast MDP. 
For the unicast MDP, the per-stage cost is $\frac{1}{K}\sum_{k\in \mathcal{K}} \phi\big(\sum_{f\in \mathcal{F}} (R_{k,f}+P_{k,f})\big)$ and the action space is $\check{\mathbf{U}}(\textbf{X})$.
By Proposition $4.2.2$ in \cite{bertsekas} and the proof of Lemma~\ref{Bellman equation}, for the unicast MDP, we learn that there exist $(\check{\theta},\check{V}(\cdot))$ satisfying:
\begin{align}\label{first stage Bellman equation}
&\check{\theta} \!+\! \check{V}(\mathbf{X})\!=\!\min \limits_{(\textbf{P}, \Delta \textbf{S})\in \check{\mathbf{U}}(\mathbf{X})} \! \Big\{ \frac{1}{K}\sum_{k\in \mathcal{K}} \phi\big(\sum_{f\in \mathcal{F}} (R_{k,f}+P_{k,f})\big)\! +\!\sum_{\textbf{A}' \in \bar{\mathcal{F}}^K}\!\!\prod_{k\in \mathcal{K}}\! q_{A_k,A_k'}^{(k)}\!\check{V}(\textbf{A}',\textbf{S}\!+\!\Delta \textbf{S}) \Big\}, \forall \textbf{X},
\end{align}
where $\check{\theta}$ and $\check{V}(\cdot)$ represent the average cost and value function of the unicast MDP, respectively. Considering that the optimal policy $\mu^*$ in (\ref{Optimal Policy1}) is a feasible policy for the unicast MDP and $\phi\big(\sum_{f\in \mathcal{F}} (R_{f}+P_{f})\big)=\frac{1}{K}\sum_{k\in \mathcal{K}} \phi\big(\sum_{f\in \mathcal{F}} (R_{k,f}+P_{k,f})\big)$ when $R_{k,f} = R_f$ and $P_{k,f} = P_f$, we have $\theta \geq \check{\theta}$.
Note that the per-stage cost of the unicast MDP is additively separable and the action space $\check{\mathbf{U}}(\mathbf{X})$ can be decoupled into $K$ local action spaces, i.e., $\check{\mathbf{U}}(\textbf{X}) = \prod_{k\in \mathcal{K}} \check{\mathbf{U}}_k(X_k)$.
Hence, $(\check{\theta},\check{V}(\mathbf{X}))$ of the unicast MDP can be expressed as $\check{V}(\mathbf{X})=\sum_{k\in\mathcal{K}} \check{V}_k(X_k)$ and $\check{\theta} = \sum_{k\in \mathcal{K}} \check{\theta}_k$, respectively, where $(\check{\theta}_k,\check{V}_k(X_k))$ satisfy:
 \begin{align}\label{lower bell per user1}
&\check{\theta}_k\! +\! \check{V}_k(X_k)\!=\!\min \limits_{(\textbf{P}_k,\Delta \textbf{S}_k )\in \check{U}_k(X_k)}\!\Big\{\!\frac{1}{K}\phi\!\big(\sum_{f\in \mathcal{F}} (\!R_{k,f}\!+\!P_{k,f})\!\big)\!+\!\sum_{A_k'\in \bar{\mathcal{F}}}q^{(k)}_{A_k,A_k'}\check{V}_k(A_k',\textbf{S}_k\!\! +\! \Delta\! \textbf{S}_k)\!\Big\},\ \forall\! X_k.
\end{align}
The Bellman equation in (\ref{lower bell per user1}) corresponds to a per-user MDP for user $k$. $\check{\theta}_k$ and $\check{V}_k(X_k)$ denote the per-user average cost and value function for user $k$, respectively. Specifically, for the per-user MDP of user $k$, at each time slot $t$, $X_k(t)=(A_k(t),\mathbf{S}_k(t))$ denotes the system state; $\mathbf{R}_k(t)$ and $\mathbf{P}_k(t)$ denote the reactive transmission action and pushing action, respectively; $\Delta \mathbf{S}_k(t)\in U_{\Delta S,k}(X_k(t), \mathbf{R}_k(t)+\mathbf{P}_k(t))$ denotes the caching action; the demand state $A_k(t)$ evolves according to the Markov chain $\{A_k(t)\}$ and the cache state $\mathbf{S}_k(t)$ evolves according to $\mathbf{S}_k(t+1) = \mathbf{S}_k(t)+ \Delta \mathbf{S}_k(t)$; $\frac{1}{K}\phi\big(\sum_{f\in \mathcal{F}} (R_{k,f}(t)+P_{k,f}(t))\big)$ denotes its per-stage cost.

Next, we establish $K$ per-user per-file MDPs via omitting the coupling among the cache units for each user of the $K$ per-user MDPs. Specifically, for each per-user MDP, considering that at each time slot, the transmission cost is either $\phi(1)$ or $0$, there is no need to push and hence each user only has to decide whether to cache the received (i.e., requested) file and evict the cached file in its storage or not. Given user state $X_k = (A_k,\mathbf{S}_k)$, we have $\Delta S_{k,f} \in \{0,-\textbf{1}(A_k \notin \{0,f\})\}$ for all $f\in \{f\in \mathcal{F}:S_{k,f}=1\}$ and $\Delta S_{k,f} = 0$ for all $f\in \{f\in \mathcal{F}:S_{k,f}=0\}$.
By omitting the constraint $\sum_{f\in \mathcal{F}} S_{k,f}+\Delta S_{k,f} = 0$ and treating each cache unit independently, we relax $U_{\Delta S,k}(X_k,\mathbf{R}\!+\!\mathbf{P})$ into $U_{\Delta S,k}'(X_k,\mathbf{R}\!+\!\mathbf{P}) \triangleq \prod_{f_k\in \mathcal{F}:S_{k,f_k}=1}U_k^1 (X_{k}^1)$, where $X_{k}^1\! \triangleq\! (A_k,f_k)$ denotes the per-user per-file state and $U_k^1(X_k^1) \triangleq \{0,-\textbf{1}(A_k \notin \{0,f_k\})\}$. 
Similarly, by Proposition $4.2.2$ in \cite{bertsekas} and the proof of Lemma~\ref{Bellman equation}, there exist $(\check{\theta}_k',\check{V}_k'(\cdot))$ satisfying:
\begin{align}\label{second stage Bellman equation1}
&\!\check{\theta}_k'\! +\! \check{V}_k'(A_k, \textbf{S}_k)= \frac{1}{K}\phi\big(\sum_{f\in \mathcal{F}} R_{k,f}\big)\!+\! \min \limits_{\Delta \textbf{S}_k\in \check{U}'_{\Delta S,k}(X_k)}\!\sum_{A_k'\in \bar{\mathcal{F}}}q^{(k)}_{A_k,A_k'}\check{V}'_k(A_k',\textbf{S}_k + \Delta \textbf{S}_k),\ \  \forall X_k,
\end{align}
where $R_{k,f}$ is given by (\ref{reactive2}). For any $X_k \in \bar{\mathcal{F}} \times \check{\mathcal{S}}$, based on the fact that $\check{U}_{\Delta S,k}(X_k) \subseteq \check{U}'_{\Delta S,k}(X_k)$, the optimal policy for the original per-user MDP, denoted as $\check{\mu}_k^*$, is feasible to the relaxed per-user MDP. Denote with $\check{\theta}_k'(\mu)$ the average cost of the relaxed per-user MDP under policy $\mu$ and then we have $\check{\theta}_k'\leq \check{\theta}_k'(\check{\mu}_k^*)=\check{\theta}_k$. Note that for per-user state $(A_k,\textbf{S}_k)$, $\frac{1}{K}\phi\big(\sum_{f\in \mathcal{F}} R_{k,f}\big)= \frac{1}{K}\phi\big(1-S_{k,A_k}\big)\textbf{1}(A_k \neq 0)= \sum_{f_k\in \mathcal{F}:S_{k,f_k}=1} \phi'(X_{k}^1)$, where $\phi'(X_{k}^1) \triangleq \frac{1}{K}\big(\frac{\phi(1)}{C}-\phi\big(\textbf{1}(A_k=f_k)\big)\big)\textbf{1}(A_k \neq 0)$,
and $\check{\textbf{U}}'_{\Delta S,k}(X_k) = \prod_{f_k\in \mathcal{F}:S_{k,f_k}=1}U_k^1(X_k^1)$. Thus, we have that
$\check{V}_k'(A_k,\textbf{S}_k)$ and $\check{\theta}_k'$ in (\ref{second stage Bellman equation1}) can be expressed as $\check{V}_k'(A_k,\textbf{S}_k)=\sum_{f\in \mathcal{F}:S_{k,f}=1} \nonumber \check{V}^1_k(A_k,f)$ and $\check{\theta}_k' = C\theta_k^1$, respectively, where for all $k\in \mathcal{K},\ X_k^1\in \bar{\mathcal F}\times\mathcal F, \ (\theta_k^1,\check{V}^1_k(X_k^1))$ satisfy (\ref{per user per content11}). Here, (\ref{per user per content11}) corresponds to the Bellman equation of the aforementioned per-user per-file MDP for user $k$ with unit cache size and $\theta_k^1$ and $\check{V}^1_k(X_k^1))$ represent the per-user per-file average cost and value function, respectively. 
 Since $C\sum_{k\in \mathcal{K}}\theta_k^1 = \sum_{k\in \mathcal{K}}\check{\theta}_k' \leq \sum_{k\in\mathcal{K}} \check{\theta}_k \leq \theta(C)$, we have $\theta(C)\geq C\sum_{k\in \mathcal{K}}\theta_k^1$. The proof ends.

\end{document}